\documentclass[journal]{IEEEtran}
\ifCLASSINFOpdf
   \usepackage[pdftex]{graphicx}
   \graphicspath{{../pdf/}{../jpeg/}}
   \DeclareGraphicsExtensions{.png,.jpeg,.png}
\else
   \usepackage[dvips]{graphicx}
   \graphicspath{{../eps/}}
   \DeclareGraphicsExtensions{.eps}
\fi

\usepackage[nolist,nohyperlinks]{acronym}
\usepackage{amsfonts}
\usepackage{amsmath}
\usepackage{amsthm}
\usepackage{color}
\usepackage{bm}
\usepackage{xcolor}
\usepackage{tabularx}
\usepackage{cite}
\usepackage[ruled,vlined]{algorithm2e}
\makeatletter
\newcommand{\nosemic}{\renewcommand{\@endalgocfline}{\relax}}
\newcommand{\dosemic}{\renewcommand{\@endalgocfline}{\algocf@endline}}
\newcommand{\norm}[1]{\left\lVert#1\right\rVert}
\newcommand{\paren}[1]{\left( #1 \right)}
\newcommand{\bracket}[1]{\left[ #1 \right]}
\newcommand{\curly}[1]{\left\{ #1 \right\}}
\newcommand{\abs}[1]{\left \lvert #1 \right \rvert}
\newcommand{\ceil}[1]{\left\lceil #1 \right\rceil}

\newtheorem{definition}{Definition}
\newtheorem{theorem}{Theorem}
\newtheorem{lemma}{Lemma}

\makeatother

\acrodef{pomdp}[Dec-POMDP]{decentralized partially observable Markov decision process}

\acrodef{scma}[SCMA]{sparse code multiple access}

\acrodef{posg}[POSG]{partially observable stochastic game}

\acrodef{mdp}[MDP]{Markov decision process}

\acrodef{cmdp}[CMDP]{constrained Markov decision process}

\acrodef{csma}[CSMA]{subcarrier sensing multiple access}

\acrodef{csma_ca}[CSMA/CA]{subcarrier sensing multiple access/collision avoidance}

\acrodef{dql}[DQL]{deep q-learning}

\acrodef{dqn}[DQN]{deep q-network}

\acrodef{dnn}[DNN]{deep neural network}

\acrodef{gfra}[GFRA]{grant-free random access}

\acrodef{ofdma}[OFDMA]{orthogonal frequency division multiple access}

\acrodef{noma}[NOMA]{non-orthogonal multiple access}

\acrodef{ofdm}[OFDM]{orthogonal frequency division multiplexing}

\acrodef{ppo}[PPO]{proximal policy optimization}

\acrodef{trpo}[TRPO]{trust region policy optimization}

\acrodef{rach}[RACH]{random access channel}

\acrodef{dnn}[DNN]{deep neural network}

\acrodef{dcqn}[DCQN]{deep convolutional Q network}

\acrodef{cnn}[CNN]{convolutional neural network}

\acrodef{iql}[IQL]{independent Q-learning}

\acrodef{ppo}[PPO]{proximal policy optimization}

\acrodef{rl}[RL]{reinforcement learning}

\acrodef{drl}[DRL]{deep reinforcement learning}

\acrodef{marl}[MARL]{multiagent reinforcement learning}

\acrodef{cert}[CERT]{concurrent experience replay trajectories}

\acrodef{bs}[BS]{base station}

\acrodef{tti}[TTI]{transmission time interval}

\acrodef{mac}[MAC]{medium access control}

\acrodef{phy}[PHY]{physical layer}

\acrodef{cw}[CW]{congestion window}

\acrodef{ber}[BER]{bit error rate}

\acrodef{plr}[PLR]{packet loss rate}

\acrodef{qam}[QAM]{quadrature amplitude modulation}

\acrodef{plr}[PLR]{packet loss rate}

\acrodef{awgn}[AWGN]{additive white Gaussian noise}

\acrodef{iid}[i.i.d.]{independent and identically distributed}

\acrodef{der}[DER]{damped exploration renewal}

\acrodef{tfe}[TFE]{trajectory feature extraction}

\acrodef{fifo}[FIFO]{first-in first-out}

\acrodef{sgd}[SGD]{stochastic gradient descent}

\acrodef{mtd}[MTD]{machine-type device}

\acrodef{mmtc}[mMTC]{massive machine-type communication}

\acrodef{phy}[PHY]{physical layer}

\acrodef{h2h}[H2H]{human-to-human}

\acrodef{iot}[IoT]{internet of things}

\acrodef{m2m}[M2M]{machine-to-machine}

\acrodef{qos}[QoS]{quality of service}

\acrodef{nb-iot}[NB-IoT]{narrowband internet of things}

\acrodef{amc}[AMC]{adaptive modulation and coding}

\acrodef{wap}[WAP]{wireless access point}

\acrodef{dpm}[DPM]{dynamic power management}

\acrodef{jal}[JAL]{joint action learners}

\acrodef{ann}[ANN]{artificial neural network}

\acrodef{gru}[GRU]{Gated Recurrent Unit}

\acrodef{lstm}[LSTM]{long short term memory}

\acrodef{lte}[LTE]{long-term evolution}

\acrodef{urllc}[URLLC]{ultra reliable low-latency communication}

\acrodef{il}[IL]{independent learners}

\acrodef{dacc}[DACC]{distributed actors with centralized critic}

\acrodef{3gpp}[3GPP]{3rd generation partnership project}

\acrodef{nr}[NR]{new radio}

\acrodef{embb}[eMBB]{enhanced mobile broadband}

\acrodef{cldi}[CLDI]{centralized learning with decentralized inference}

\acrodef{ack}[ACK]{acknowledgement}

\acrodef{nack}[NACK]{negative-acknowledgement}

\acrodef{adam}[ADAM]{adaptive moment estimation}

\acrodef{vr}[VR]{virtual reality}

\acrodef{ar}[AR]{augmented reality}

\acrodef{v2x}[V2X]{vehicle-to-everything}

\acrodef{ciot}[cIoT]{critical internet of things}

\acrodef{kpi}[KPI]{key performance indicator}

\acrodef{lafp}[LAFP]{latent access failure probability}

\acrodef{sps}[SPS]{semi-persistent scheduling}

\acrodef{ue}[UE]{user equipment}

\acrodef{ra}[RA]{random access}

\acrodef{mimo}[MIMO]{multiple input multiple output}

\acrodef{harq}[HARQ]{hybrid automatic repeat request}

\acrodef{mmwave}[mmWave]{millimeter wave}

\acrodef{hppp}[HPPP]{homogeneous Poisson point process}

\acrodef{ula}[ULA]{uniform linear array}

\acrodef{los}[LOS]{line-of-sight}

\acrodef{nlos}[NLOS]{non-line-of-sight}

\acrodef{prach}[PRACH]{physical random access channel}

\acrodef{rar}[RAR]{random access request}

\acrodef{pusch}[PUSCH]{physical uplink shared channel}

\acrodef{csi}[CSI]{channel state information}

\acrodef{sinr}[SINR]{signal-to-interference-plus-noise-ratio}

\acrodef{ack}[ACK]{acknowledgement}

\acrodef{nack}[NACK]{negative acknowledgement}

\acrodef{rtt}[RTT]{round-trip time}

\acrodef{qos}[QoS]{quality of service}

\acrodef{pgfl}[PGFL]{probability generating functional}

\acrodef{pmf}[PMF]{probability mass function}

\acrodef{fr2}[FR2]{frequency range 2}

\acrodef{ccdf}[CCDF]{complementary cumulative distribution function}

\begin{document}
%
\title{Reliability and User-Plane Latency Analysis of mmWave Massive MIMO for Grant-Free URLLC Applications}
%
%
%

\author{Joao V. C. Evangelista,~\IEEEmembership{Student Member,~IEEE,}
        Georges Kaddoum,~\IEEEmembership{Senior Member,~IEEE}, Zeeshan Sattar,~\IEEEmembership{Member,~IEEE}
\thanks{\ifdefined \arxiv
This work has been submitted to the IEEE for possible publication.  Copyright may be transferred without notice, after which this version may no longer be accessible.

\else
\fi
J. V. C. Evangelista  and G. Kaddoum were with the Department
of Electrical Engineering, Ecole de Technologie Supérieure, Montreal,
QC, H3C 1K3 CA, e-mail: joao-victor.de-carvalho-evangelista.1@ens.etsmtl.ca and georges.kaddoum@etsmtl.ca.
Z. Sattar was with Ericsson Canada, Ottawa, ON, K2K 2V6 CA, email:  zeeshan.sattar@ericsson.com}
}
\maketitle

\begin{abstract}
5G cellular networks are designed to support a new range of applications not supported by previous standards. Among these, ultra-reliable low-latency communication (URLLC) applications are arguably the most challenging. URLLC service requires the user equipment (UE) to be able to transmit its data under strict latency constraints with high reliability. To address these requirements, new technologies, such as mini-slots, semi-persistent scheduling and grant-free access were introduced in 5G standards. In this work, we formulate a spatiotemporal mathematical model to evaluate the user-plane latency and reliability performance of millimetre wave (mmWave) massive multiple-input multiple-output (MIMO) URLLC with reactive and $K$-repetition hybrid automatic repeat request (HARQ) protocols. We derive closed-form approximate expressions for the latent access failure probability and validate them using numerical simulations. The results show that, under certain conditions, mmWave massive MIMO can reduce the failure probability by a factor of $32$. Moreover, we identify that beyond a certain number of antennas there is no significant improvement in reliability. Finally, we conclude that mmWave massive MIMO alone is not enough to provide the performance guarantees required by the most stringent URLLC applications.
\end{abstract}

%
\IEEEpeerreviewmaketitle

%
%

\section{Introduction}
\label{sec:intro}

\IEEEPARstart{T}{he} \ac{3gpp} has identified three distinct use cases for 5G \ac{nr} and beyond cellular networks based on their different connectivity requirements: \ac{embb}, \ac{mmtc} and \ac{urllc} \cite{series2015imt}. Since the inception of the idea of 5G \ac{nr}, it has been argued that its main revolution is a change of paradigm from a smartphone-centric network to a network capable of satisfying the requirements of diverse services, such as machine-to-machine and vehicle-to-vehicle communications \cite{WirelessFuture2021}. The \ac{urllc} scenario targets applications that require high reliability and low latency, such as \ac{ar}, \ac{vr}, \ac{v2x}, \ac{ciot}, industrial automation and healthcare. According to the use cases defined in \cite{3gpp.38.913}, the main \ac{kpi} to be satisfied in \ac{urllc} applications is the latent access failure probability, which incorporates the reliability and latency requirements needed in such applications. The requirements for \ac{urllc} applications vary from $1-10^{-5}$ transmission reliability to transmit $32$ bytes of data with a user-plane latency of less than $1$ ms to a $1-10^{-5}$ reliability to transmit $300$ bytes with a user-plane latency of between $3$ and $10$ ms, depending on the application \cite{3gpp.38.913}. 

\Ac{lte} and prior networks were not designed with such constraints in mind. Scheduling in \ac{lte} follows a grant-based approach, where the \ac{ue} must request resources in a $4$-step \ac{ra} procedure before transmitting data \cite{Vilgelm2018}. In the best-case scenario, it takes at least $10$ ms for a \ac{ue} to start transmitting its payload. 

Therefore, new mechanisms were introduced into the 5G \ac{nr} specification to support the latency requirements of \ac{urllc} applications. Firstly, a flexible numerology was proposed, introducing the concept of a mini-slot that last as little as $0.125$ ms \cite{3gpp.38.912}, in contrast to the $1$ ms minimum slot duration on \ac{lte}, enabling fine-grained scheduling of network resources \cite{Zaidi2018}. Secondly, the introduction of \ac{sps} of grants \cite{3gpp.38.213, Karadag2019}, where some of the networkś resource blocks are periodically reserved for \ac{urllc} applications, thereby avoiding the grant request procedure. Despite the efforts, not all \ac{urllc} applications have a periodic traffic pattern and are therefore unable to benefit greatly from \ac{sps}. Additionally, some services require low latency and reliable transmission to transmit small sporadic packets. With that in mind, both the standards committee \cite{3gpp.1704222} and researchers have put a lot of effort to investigate grant-free transmission, where the \acp{ue} transmit their payload directly in the \ac{ra} channel. This culminated with the introduction of the $2$-step \ac{ra} procedure introduced in Release 16 \cite{3gpp.1704222}. The $2$-step \ac{ra} procedure follows a grant-free approach, where instead of waiting for a dedicated channel to be assigned by the network, it transmits its data directly into the \ac{ra} channel and waits for feedback from the network \cite{Kim2021}. 

Moreover, massive \ac{mimo} is a fundamental part of 5G \ac{nr} \cite{Lu2014, Larsson2014, Bjornson2018}. It provides performance gains by improving diversity against fading and, along with advanced signal processing techniques, can provide directivity to transmission/reception, mitigating interference between spatially uncorrelated \acp{ue} \cite{Marzetta2016}. The performance enhancements provided by \ac{mimo} are essential to ensure the reliability and the low latency required by \ac{urllc} applications. In conjunction with massive \ac{mimo} comes \ac{mmwave} transmission. Due to its small wavelength, \ac{mmwave} antennas can be packed into massive arrays, making it a key enabler of massive \ac{mimo} systems which attracted a significant interest on the topic \cite{Rappaport2013, Rangan2014, Andrews2016, Sattar2017, Sattar2019a, Sattar2019b}. However, \ac{mmwave} propagation comes with its own challenges due to the severe propagation loss experienced by electromagnetic signals in this frequency range.

In this paper, we develop a spatiotemporal analytical model to evaluate the performance of \ac{mmwave} massive \ac{mimo} communication systems for \ac{urllc} applications. We use tools from stochastic geometry and probability theory to evaluate and compare system performance metrics by deriving closed-form approximate expressions for its latent access failure probability under different \ac{harq} protocols.

\subsection{Related Work}
\label{sec:lit_rev}

In \cite{Gao2019}, the authors propose a queueing model to compare the throughput performance of packet-based (grant-free) and connection-based (grant-based) random access. They conclude that packet-based systems with sensing can achieve greater throughput than connection-based one for small packet transmissions. In \cite{Liu2021, Evangelista2021}, the optimization of grant-free access networks is investigated. The former considers the dynamic optimization of \ac{harq} and scheduling parameters with \ac{noma}, while the former considers the distributed link adaptation problem. Both papers formulate the respective optimization tasks as \ac{marl} problems. The probability of success and the area spectral efficiency of a grant-free \ac{scma} system is evaluated in \cite{Lai2021, Evangelista2019b}, in an \ac{mmtc} context, using stochastic geometry. However, none of the works consider the temporal aspects of the system, which are crucial to analyze the latency and reliability of \ac{urllc} service. In \cite{Ding2019}, the probability of success of grant-free \ac{ra} with massive \ac{mimo} in the sub-6 GHz band is investigated, and analytical expressions are derived for conjugate and zero-forcing beamforming. Despite its contribution, the authors do not evaluate the systemś temporal behavior, which is fundamental to characterize \ac{urllc} service's performance. Moreover, due to its distinct propagation characteristics, this model is unsuitable for \ac{mmwave} frequency bands.

The authors in \cite{Gharbieh2018} evaluate the scalability of scheduled uplink (grant-based) and random access (grant-free) transmissions in massive \ac{iot} networks, although they frame the problem through a revolutionary spatiotemporal framework, fusing stochastic geometry and queueing theory. They conclude that grant-free transmission offers lower latency, however, it does not scale well to a massive number of devices. In our work, we show that using massive \ac{mimo} \acp{bs} is a viable solution to address the scalability issues of grant-free transmission without sacrificing its latency, rendering it particularly suitable for \ac{urllc} applications. In \cite{Jiang2018}, the authors use a similar spatiotemporal model to characterize the performance of different \ac{ra} schemes with respect to the probability of a successful preamble transmission in a grant-based massive \ac{iot} system. They conclude that a backoff scheme performs close to optimally in diverse traffic conditions. In \cite{Jacobsen2017}, the authors perform system-level simulations of a grant-free \ac{urllc} network under different \ac{harq} configurations, and compare it to a baseline grant-based system. They conclude that grant-free systems provide significantly lower latency at the $1 - 10^{-5}$ reliability level. The same scenario is evaluated in \cite{Liu2020}, however, the authors characterize system performance analytically, using a stochastic geometry-based spatiotemporal model. This paper identifies the suitability of each \ac{harq} scheme for different network loads and received power levels.

Stochastic geometry has become the de facto tool for analyzing large networks \cite{Lu2021, Hmamouche2021, Jiang2018b}, and has been successfully used to investigate the performance of \ac{mimo} systems for a while now \cite{Tanbourgi2015, Nguyen2013, Adhikary2015, Lee2014}. In \cite{Afify2016}, a unified stochastic geometric mathematical model for \ac{mimo} cellular networks with retransmission is proposed. In \cite{Ding2017}, a stochastic geometry-based analytical model for the performance of downlink \ac{mmwave} \ac{noma} systems is developed. The authors propose two random beamforming methods that are able to reduce system overhead while providing performance gains for \acp{bs} with a large number of antennas.

We seek to answer the following main questions that are to the best of our knowledge missing from the current literature:
\begin{itemize}
    \item How do we formulate a tractable spatiotemporal model to investigate the reliability and latency of \ac{urllc} applications powered by \acp{bs} equipped with massive antenna arrays operating on \ac{mmwave} frequencies?
    \item What closed-form analytical expressions can we derive for the latent access failure probability in this scenario?
    \item What are the performance gains obtained from increasing the number of antennas at the \ac{bs}, and what are the limitations?
\end{itemize}

\subsection{Contributions}
\label{sec:contributions}

This paper makes three major contributions:

\begin{itemize}
    \item We formulate a mathematical model to evaluate the performance of \ac{mmwave} massive \ac{mimo} on uplink grant-free \ac{urllc} networks with \ac{harq}. This model uses stochastic geometry to capture the spatial configuration of the \acp{ue} and the \acp{bs}, a \ac{mmwave} channel model, and probability theory to obtain the temporal characteristics necessary to evaluate the performance of \ac{urllc} applications. 
    \item We derive closed-form approximate expressions for the latent access failure probability using reactive and $K$-repetition \ac{harq} schemes. To the best of our knowledge, no previous works has presented closed-form analytic expressions for this key performance measure of \ac{urllc} applications in a \ac{mmwave} massive \ac{mimo} communication system.
    \item We analyze the system performance for an extensive range of scenarios, identifying the gains and limitations provided by using the \ac{mmwave} spectrum together with a massive number of antennas at the \ac{bs}, and identify the scenarios that benefit the most from these two technologies.
\end{itemize}

\subsection{Notation and Organization}
\label{sec:notation}
Italic Roman and Greek letters denote deterministic and random variables, while bold letters denote deterministic and random vectors. The capital Greek letter $\Phi$ denotes a point process and $x \in \Phi$ represents a point belonging to said process. The notation $\Phi(A)$, where $A \in \mathbb{R}^d$, is the counting process associated with $\Phi$ \cite{Haenggi2012}. Notice that we overload the meaning of $\Phi$ so that it can signify a point process, a counting measure or a set depending on the context. ${n \choose k} = \frac{n!}{k! \paren{n - k}!}$ is the binomial coefficient of $n$ choose $k$.

The uniform, complex normal and binomial distributions are represented by $Uniform(a,b)$, $\mathcal{CN}(\mu, \sigma^2)$ and $Binomial(p)$, respectively. The vector $\mathbf{x}^H$ is the Hermitian transpose of vector $\mathbf{x}$. The function $\mathbb{P}(\cdot)$ denotes the probability of the event within parentheses. The notation $\mathbf{1} \curly{\cdot}$ denotes the indicator function, which is equal to one whenever the event within curly braces is true and zero otherwise.

This work is divided into five sections and an appendix. In Section \ref{sec:intro}, we introduce the contents of the manuscript and contextualize within the relevant literature. In Section \ref{sec:system_model}, we present a mathematical model to characterize the performance of the grant-free \ac{mmwave} massive \ac{mimo} system in \ac{urllc} applications. In Section \ref{sec:system_analysis}, we derive the latent access failure probability of the proposed system using reactive and $K$-repetition \ac{harq} protocols. In Section \ref{sec:numerical_results}, we show the results of the system simulation. We use these results to validate the analytical derivations, investigate the system's performance for an extensive range of parameters and finish it by interpreting the results in the context of \ac{urllc} applications. In Section \ref{sec:conclusions}, we summarize our findings and present our conclusions. Finally, in the appendix, we show the detailed proofs of the lemmas and theorems required by the derivations in the paper.

%
%

\section{System Model}
\label{sec:system_model}

\ifCLASSOPTIONtwocolumn
\begin{figure*}
    \centering {
    \includegraphics[width=\columnwidth]{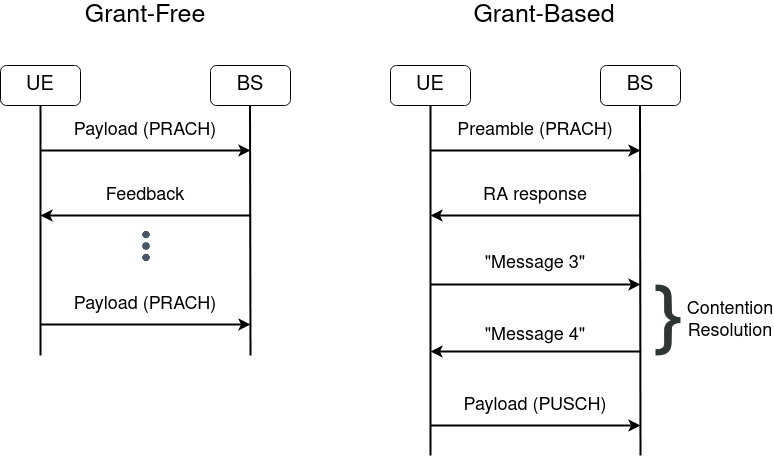}
    \caption{Comparison of the transmission procedure in grant-free and grant-based systems.}
    \label{fig:gf_vs_gb}
    }
\end{figure*}
\else
\begin{figure}
    \centering {
    \includegraphics[width=0.6\columnwidth]{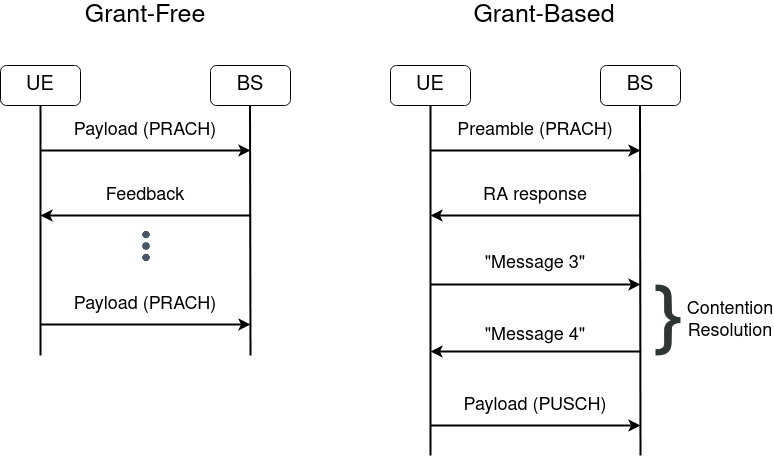}
    \caption{Comparison of the transmission procedure in grant-free and grant-based systems.}
    \label{fig:gf_vs_gb}
    }
\end{figure}
\fi

In most cellular applications, uplink transmissions use a dedicated resource (frequency, time or a \ac{mimo} spatial layer) previously assigned by the network to transmit their data payload. Thus, when an \ac{ue} receives new data, it must send a request for the network to schedule a resource. With dedicated resources, each \ac{ue} can utilize the wireless channel to its full capacity, thus maintaining good \ac{qos}. TIn 5G \ac{nr} networks, the schedule request consists of four steps, illustrated in Figure \ref{fig:gf_vs_gb}:
\begin{itemize}
    \item The \ac{ue} randomly selects one of the available preambles and transmits it on the \ac{prach}.
    \item The \ac{bs} transmits a \ac{rar}, acknowledging receipt of the preamble and time-alignment commands.
    \item The \ac{ue} and \ac{bs} exchange contention resolution messages (messages 3 and 4) that are used to identify possible collisions arising from two different devices transmitting the same preamble. 
    \item If the grant request is successful, the \ac{ue} transmits its payload on the \ac{pusch}.
\end{itemize}
This grant-based scheme is efficient for applications that need to use the channel multiple times to transmit large amounts of data (e.g., video streaming) or data that's being continuously generated (e.g., voice). However, in some \ac{urllc} applications, \acp{ue} sporadically generate data that need to be transmitted reliably and with low latency, such as \ac{ciot} and sensors for industrial automation. In such scenarios, the time spent on the schedule request renders grant-based schemes inefficient. A more suitable alternative is to transmit the data directly on the \ac{prach} and thereby avoidall the overhead involved in requesting a grant, as illustrated in Fig. \ref{fig:gf_vs_gb}. Nonetheless, with grant-free transmission comes the possibility of collisions whenever two \acp{ue} randomly select the same preambles. Therefore, \ac{harq} is used to ensure the reliability and robustness of grant-free transmission. \Ac{harq} consists of using feedback information from the \ac{bs} so the \ac{ue} can retransmit packets that were not successfully received. Despite this, it can be quite challenging to scale grant-free networks because wireless resources are finite and expensive. To this end, massive \ac{mimo} and beamforming can be applied to reduce the interference of spatially uncorrelated \acp{ue} and thereby increasing the reliability of the system.

In this section, we discuss the spatial model of the network, the \ac{mmwave} channel model, the \ac{bs} receiver beamforming procedure and the different \ac{harq} schemes used.

\subsection{Physical Layer Model}
\label{sec:spatial_model}
Stochastic geometry and the theory of random point processes has proven to be
able to accurately model the spatial distribution of modern cellular network deployments \cite{Lu2015}. Therefore, we consider a cell of radius $R$ consisting of a \ac{bs}, equipped with $K$ antennas, located at the origin. We model the spatial location of the single-antenna \acp{ue} according to a \ac{hppp} \cite{Haenggi2012}, denoted by $\Phi_U$ with intensity $\lambda_U$. Furthermore, the distance between the $i$-th \ac{ue}, $x_i \in \Phi_U$, and the \ac{bs} is given by $\norm{x_i}$. Both the distance from the \ac{ue} to the \ac{bs} and its normalized angle from the \ac{bs} are uniformly distributed random variables \cite{Haenggi2012}, $\norm{x_i} \sim Uniform(0, R)$ and $\theta_i \sim Uniform(-1, 1)$, respectively. 

Due to path loss attenuation, the signal received from \acp{ue} located further from the \ac{bs} is ``drowned'' by the signal from closer users transmitting with the same power, also known as the near-far problem. Uplink power control is fundamental to deal with this issue. We consider that the \acp{ue} utilize path loss inversion power control \cite{Elsawy2014}, with received power threshold $\rho$, where each user controls its transmit power such that the average received power at its associated \ac{bs} is $\rho$, by selecting their transmit powers as $p_i = \rho \norm{x_i}^{\alpha}$, where $\alpha$ is the path loss exponent. We assume that there are $N_S$ subcarriers reserved for grant-free \ac{urllc} transmissions and $N_P$ orthogonal preambles. Thus, at each \ac{tti}, the active \acp{ue} select a subcarrier and preamble randomly from the $N_S$ available subcarriers and $N_P$ available preambles. Moreover, we assume that at $t=0$, one packet arrives to the transmitting queue of each \ac{ue}. Therefore, the \ac{hppp} of active users $\Phi_A$ on a specific subcarrier is obtained by thinning $\Phi_U$ \cite{Haenggi2012} and its effective intensity at $t=0$ is given by
\begin{equation}
    \label{eq:active_lambda}
    \lambda_A = \frac{\lambda_U}{N_S}.
\end{equation}

Massive \ac{mimo} technology and \ac{mmwave} frequencies are intrinsically connected. Even though one does not imply the other, they complement each other really well. The former requires large antenna arrays, and the size of such arrays is proportional to the targeted wavelength. Moreover, \ac{mmwave} antennas must be really small to operate in such large frequencies, therefore, a larger number of them are necessary to gather enough energy. In this work, we consider that the \ac{bs} is equipped with a massive \ac{ula} containing $K \gg 1$ antennas operating at \ac{mmwave} frequencies, while the \acp{ue} possess a single antenna. The channel vector between user $i$ and the \Ac{bs} is given by
\begin{equation}
    \label{eq:channel_gain}
    \mathbf{h}_i = \sqrt{K} \left[ \underbrace{\frac{g_{i, 0} \mathbf{a}(\theta_i^0)}{\sqrt{\norm{x_i}^{\alpha_{LOS}}}}}_{\text{LOS component}} + \underbrace{\frac{\underset{j=1}{\overset{J}{\sum}} g_{i,j} \mathbf{a}(\theta_i^j)}{\sqrt{\norm{x_i}^{\alpha_{NLOS}}}}}_{\text{NLOS components}} \right],
\end{equation}
where $g_{i,j} \sim \mathcal{CN}(0, 1)$ is the complex gain on the $j$-th path and $\theta^j_i$ is the normalized direction of the $j$-th path. We assume that the complex gains of different paths are independent. $\alpha_{LOS}$ and $\alpha_{NLOS}$ denote the path loss exponent of the \ac{los} and \ac{nlos} paths, respectively. The vector
\begin{equation}
    \label{eq:phase_array}
    \mathbf{a}(\theta) = \frac{1}{\sqrt{K}}
    \begin{bmatrix}
        1 & e^{-j \pi \theta} & \dots & e^{-j \pi (K-1) \theta}
    \end{bmatrix}^T
\end{equation}
denotes the phase of the signal received by each antenna. Due to high penetration losses suffered by \ac{mmwave} signals, the \ac{los} path has a dominant effect on channel gain, being $20$ dB larger than the \ac{nlos} in some cases \cite{Ding2017, Lee2014}. Hence, we can safely approximate $\mathbf{h}_i $ as
\begin{equation}
    \label{eq:approx_channel}
    \mathbf{h}_i \approx \sqrt{K} \frac{g_{i} \mathbf{a}(\theta_i)}{\sqrt{\norm{x_i}^{\alpha}}}
\end{equation}
for mathematical tractability. Additionally, to avoid cluttering the notation, we drop the subscripts denoting different paths and distinguishing \ac{los} and \ac{nlos} variables. 

Due to the dominant effect of the \ac{los} link, the channel model also needs to consider a blockage model to determine the probability that the \ac{los} path between the \ac{ue} the \ac{bs} is obstructed. To model the effects of blockage, we adopt the model proposed in \cite{Thornburg2016}. This model is obtained by assuming that the obstructing building and structures form an \ac{hppp} with random width, length and orientation. So, let $LOS$ be the set of \Ac{los} \acp{ue}; then, the probability that user $x_i$ has a \ac{los} link is given by
\begin{equation}
    \label{eq:los_prob}
    \mathbb{P} \paren{x_i \in LOS} = \exp \paren{- \beta \norm{x_i}},
\end{equation}
where $\beta$ is directly proportional to the density, and the average width and length of obstructing structures. This model nicely captures the exponentially vanishing probability of having a \ac{los} link the further you move away from the \ac{bs}, and can be easily fitted to real urban scenarios.

\subsubsection*{Signal Model}
\label{sec:signal_model}
At each \ac{tti}, the active users transmit an information signal $s_i$ such that $\abs{s_i} = 1$. Therefore, the vector of the signal received at the \ac{bs} is given by
\begin{equation}
    \mathbf{y} = \underset{x_i \in \Phi_A}{\sum} \mathbf{1}\curly{x_i \in LOS} \sqrt{\rho} \mathbf{h}_i s_i + \mathbf{n},
\end{equation}
where $\mathbf{n} \sim \mathcal{CN} \paren{0, \sigma_2 \mathbf{I}}$ is a circularly symmetric complex Gaussian random variable representing \ac{awgn}.

To successfully recover the data transmitted by a given user, the \ac{bs} must be able to accurately estimate its channel response. 
\begin{definition}[Preamble Collision]
    Preamble collision event, denoted by $C$, happens when two or more devices transmit the same preamble on the same subcarrier.
\end{definition}
We assume that the \ac{bs} is able to perfectly estimate the \ac{ue} channel response $\mathbf{h}_i$ whenever there is no preamble collision. Then, the \ac{bs} performs conjugate beamforming to separate the intended user's signal from those of the other interfering \acp{ue} by multiplying the received signal by the Hermitian transpose of the intended user channel response. Therefore the recovered signal of the intended user, the $j$-th user, is
\ifCLASSOPTIONtwocolumn
\begin{eqnarray}
    y_j &=&  \mathbf{1}\curly{x_j \in LOS} \mathbf{1}\curly{\bar{C}} \sqrt{\rho} \mathbf{h}_j^H \mathbf{h}_j s_j + \nonumber \\ && \underset{x_i \in \Phi_A \setminus \curly{x_j}}{\sum} \mathbf{1}\curly{x_i \in LOS} \sqrt{\rho} \mathbf{h}_j^H \mathbf{h}_i s_i + \tilde{n}_j, \nonumber \\
\end{eqnarray}
\else
\begin{eqnarray}
    y_j &=&  \mathbf{1}\curly{x_j \in LOS} \mathbf{1}\curly{\bar{C}} \sqrt{\rho} \mathbf{h}_j^H \mathbf{h}_j s_j + \nonumber \\ && \underset{x_i \in \Phi_A \setminus \curly{x_j}}{\sum} \mathbf{1}\curly{x_i \in LOS} \sqrt{\rho} \mathbf{h}_j^H \mathbf{h}_i s_i + \tilde{n}_j,
\end{eqnarray}
\fi
where $\bar{C}$ is the event when user $j$ does not experience preamble collision and $\tilde{n}_j \sim \mathcal{CN} \paren{0, \sigma^2}$ is a linear combination of the noise vector, which is a Gaussian distributed random variable. Therefore, the \ac{sinr} experienced by user $j$ is given by
\ifCLASSOPTIONtwocolumn
\begin{eqnarray}
    \label{eq:sinr}
    &&\mathrm{SINR} \nonumber \\ &=& \frac{\mathbf{1}\curly{x_j \in LOS \cap \bar{C}} \rho \abs{\mathbf{h}_j^H \mathbf{h}_j}^2}{\underset{x_i \in \Phi_A \setminus \curly{x_j}}{\sum} \mathbf{1}\curly{x_i \in LOS} \rho \abs{\mathbf{h}_j^H \mathbf{h}_i}^2 + \sigma^2} \nonumber \\ 
    &=& \frac{\mathbf{1}\curly{TX_j} \rho \abs{g_j}^2 \abs{\mathbf{a}\paren{\theta_j}^H \mathbf{a}\paren{\theta_i}}^2}{I + \sigma^2}, 
\end{eqnarray}
\else
\begin{eqnarray}
    \label{eq:sinr}
    \mathrm{SINR} &=& \frac{\mathbf{1}\curly{x_j \in LOS \cap \bar{C}} \rho \abs{\mathbf{h}_j^H \mathbf{h}_j}^2}{\underset{x_i \in \Phi_A \setminus \curly{x_j}}{\sum} \mathbf{1}\curly{x_i \in LOS} \rho \abs{\mathbf{h}_j^H \mathbf{h}_i}^2 + \sigma^2} \nonumber \\ 
    &=& \frac{\mathbf{1}\curly{TX_j} \rho \abs{g_j}^2 \abs{\mathbf{a}\paren{\theta_j}^H \mathbf{a}\paren{\theta_i}}^2}{I + \sigma^2}, 
\end{eqnarray}
\fi
where $TX_j = \{x_j \in LOS \cap \bar{C}\}$ is the probability that user $j$ has a \ac{los} link and does not suffer from preamble collision, and $I = \underset{x_i \in \Phi_A \setminus \curly{x_j}}{\sum} \mathbf{1}\curly{x_i \in LOS} \rho \abs{g_i}^2 \abs{\mathbf{a}\paren{\theta_j}^H \mathbf{a}\paren{\theta_i}}^2$ is the interference from the other \acp{ue}. Moreover, the beamforming gain, $\abs{\mathbf{a}\paren{\theta_j}^H \mathbf{a}\paren{\theta_i}}^2$, can be expressed as \cite{Lee2015}
\ifCLASSOPTIONtwocolumn
\begin{equation}
\begin{split}
    \label{eq:fejer_kernel}
    \abs{\mathbf{a}\paren{\theta_j}^H \mathbf{a}\paren{\theta_i}}^2 &= F_K\paren{\frac{\pi}{2} \paren{\theta_i - \theta_j}} \\ 
    &= \frac{1}{K} \abs{\frac{\sin \paren{\frac{K \pi}{2} \paren{\theta_i - \theta_j}}}{\sin \paren{\frac{\pi}{2} \paren{\theta_i - \theta_j}}}}^2,
\end{split}
\end{equation}
\else
\begin{equation}
    \label{eq:fejer_kernel}
    \abs{\mathbf{a}\paren{\theta_j}^H \mathbf{a}\paren{\theta_i}}^2 = F_K\paren{\frac{\pi}{2} \paren{\theta_i - \theta_j}} = \frac{1}{K} \abs{\frac{\sin \paren{\frac{K \pi}{2} \paren{\theta_i - \theta_j}}}{\sin \paren{\frac{\pi}{2} \paren{\theta_i - \theta_j}}}}^2,
\end{equation}
\fi
where $F_K(x)$ is the Fejer kernel \cite{Marsden1993}, with $F_K(0) = K$.
\ifCLASSOPTIONtwocolumn
\begin{figure}
    \centering {
    \includegraphics[width=\columnwidth]{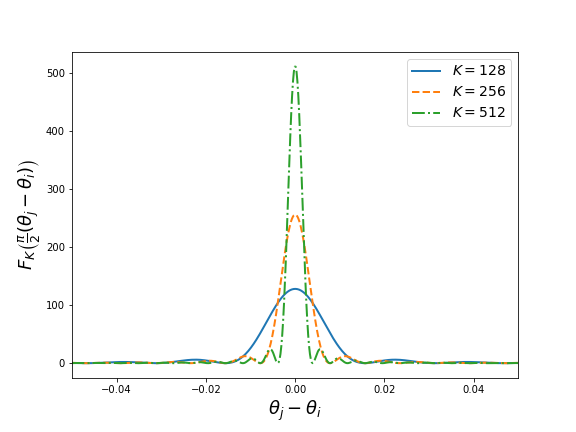}
    \caption{Fejer kernel value for normalized angles of arrival varying from $-0.2$ to $0.2$.}
    \label{fig:fejer_kernel}
    }
\end{figure}
\else
\begin{figure}
    \centering {
    \includegraphics[width=.6\columnwidth]{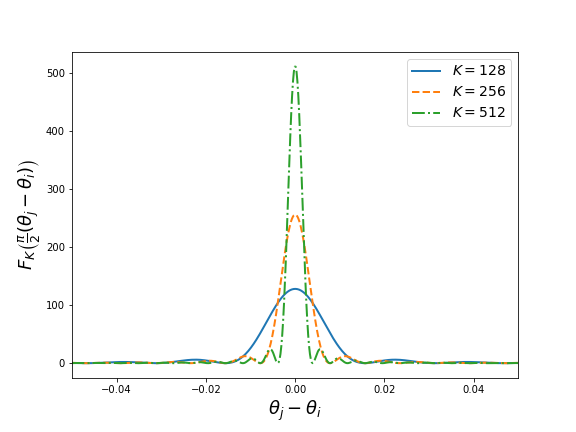}
    \caption{Fejer kernel value for normalized angles of arrival varying from $-0.05$ to $0.05$.}
    \label{fig:fejer_kernel}
    }
\end{figure}
\fi
A useful property of the Fejer kernel is that \cite{Marsden1993}
\begin{equation}
    \label{eq:fejer_property}
    \underset{K \rightarrow \infty}{\lim} \int_{\delta \leq \abs{x} \leq \pi} F_K(x) dx = 0,
\end{equation}
meaning that for an asymptotically large value of $K$, the interference for the signals not aligned with beam angle $\theta_j$ goes to zero. Fig. \ref{fig:fejer_kernel} illustrates this property by plotting the Fejer kernel for increasing values of $K$.

\subsection{HARQ Schemes}

\ac{harq} protocols determine how transmitters and receivers exchange information about successful packet reception, by transmitting an \ac{ack} signal, and how \acp{ue} retransmit in the event of failure, which is signaled by the transmission of a \ac{nack} signal. They are especially important to ensure reliability in grant-free transmission. The \ac{harq} protocol used also impacts the overall latency of the system. Hence, in this paper, we investigate the performance of the massive \ac{mimo} \ac{urllc} network under two distinct \ac{harq} protocols.

With respect to transmissions latency, the \ac{harq} protocols investigated have a few aspects in common. First, the \ac{ue} spends $T_A$ \acp{tti} to process a newly arrived packet. As soon as the packet is processed, it spends $T_{TX}$ \acp{tti} transmitting it. Upon receipt of the packet, the \ac{bs} spends $T_{DP}$ \acp{tti} to process it and $T_F$ \acp{tti} to send feedback and for it to reach the \ac{ue}. Once the \ac{ue} receives the feedback signal, it takes $T_{UP}$ \acp{tti} to process it. We consider that the transmit and feedback time already take into account the propagation delay between the transmitter and receiver. Without loss of generality, we assume that $T_A = T_{TX} = T_{DP} = T_{UP} = 1$ \ac{tti}. Another concept shared between different \ac{harq} protocols is the \ac{rtt}, which consists of the time it takes from the start of a transmission by the \ac{ue} to the end of processing of the feedback signal, either \ac{ack} or \ac{nack}, the \ac{ue} received from the \ac{bs}.

\ifCLASSOPTIONtwocolumn
\begin{figure}
    \centering {
    \includegraphics[width=.8\columnwidth]{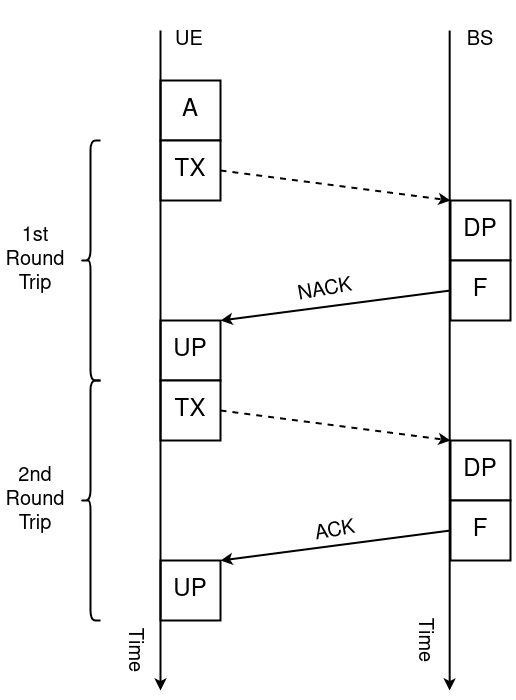}
    \caption{An illustration of a couple of reactive HARQ protocol round trips.}
    \label{fig:reactive}
    }
\end{figure}

\begin{figure}
    \centering {
    \includegraphics[width=.8\columnwidth]{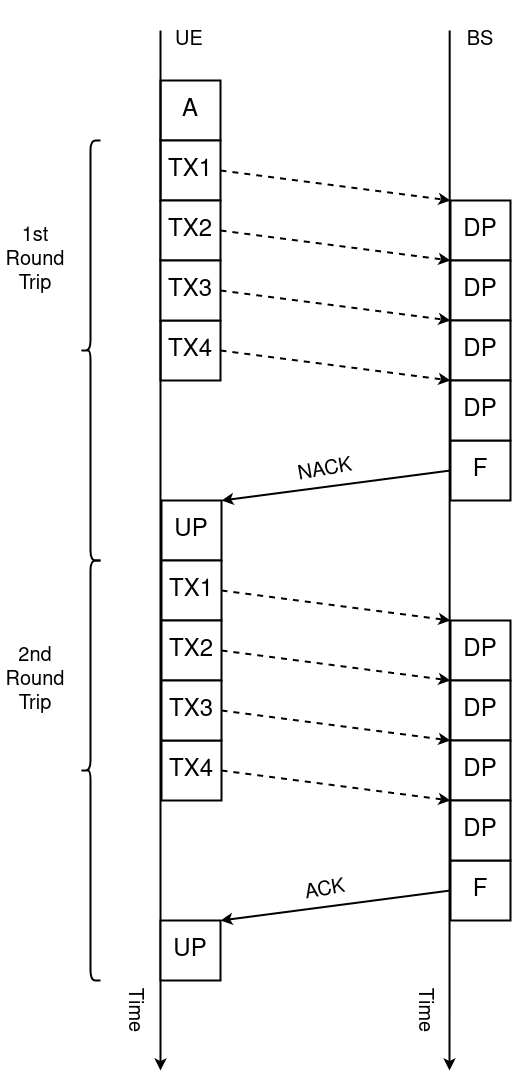}
    \caption{An illustration of a couple of $K$-repetition HARQ protocol round trips.}
    \label{fig:krep}
    }
\end{figure}

\else
\begin{figure}
    \begin{minipage}{.4\textwidth}
        \centering {
        \includegraphics[width=\columnwidth]{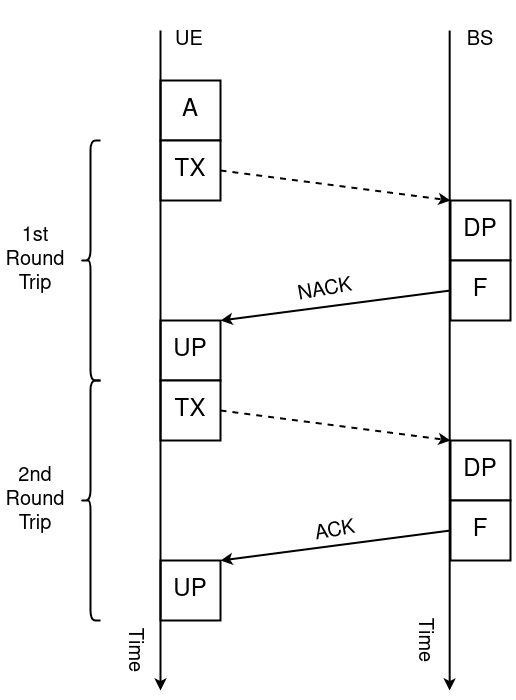}
        \caption{An illustration of a couple of reactive HARQ protocol round trips.}
        \label{fig:reactive}
        }
    \end{minipage}
    \begin{minipage}{.4\textwidth}
        \centering {
        \includegraphics[width=\columnwidth]{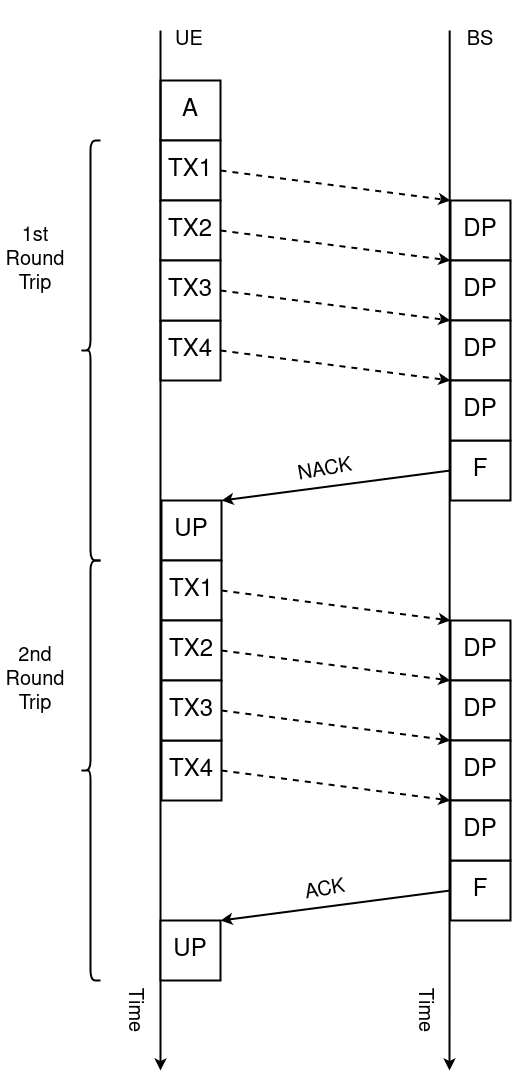}
        \caption{An illustration of a couple of $K$-repetition HARQ protocol round trips.}
        \label{fig:krep}
        }
    \end{minipage}
%
\end{figure}
\fi

\subsubsection{Reactive Scheme}
\label{sec:reactive_harq}
The reactive \ac{harq} protocol is the more straight-forward one of the two considered in this paper. The \ac{ue} attempts to transmit one packet and waits for feedback from the \ac{bs}. Once the feedback is processed, it either attempts to retransmit the same packet if it got a \ac{nack} signal or sits idle until a new packet arrives. This protocol is illustrated in Fig. \ref{fig:reactive}, which shows the processing times and signal exchange between the \ac{ue} and the \ac{bs}. Under the assumptions considered in this paper, the reactive \ac{rtt} $T^{reac}_{RTT}$ is given by
\begin{equation}
    \label{eq:reactive_rtt}
    T^{reac}_{RTT} = 4 \text{ TTIs}.
\end{equation}
From (\ref{eq:reactive_rtt}), the user-plane latency of the $m$-th \ac{harq} round-trip is
\begin{equation}
    \label{eq:reactive_latency}
    T^{reac}(m) = T_A + 4 m \text{ TTIs}.
\end{equation}

\subsubsection{$K$-Repetition Scheme}
\label{sec:krep_harq}
To increase the reliability and robustness of each transmission attempt, the $K$-repetition \ac{harq} protocol repeats the same packet $K_{rep}$ times on each attempt. Therefore, the only way a transmission attempt fails is if each of the $K_{rep}$ transmissions fail, which translates into an increased reliability of the overall system. However, feedback on the transmission attempt is sent only after the last repetition is processed by the \ac{bs}. So, there is a tradeoff between enhancing the reliability of each transmission and increasing the latency of a transmition attempt. Fig. \ref{fig:krep} shows two $K$-repetition round-trip transmissions, where the first transmission fails and the second is successful. The \ac{rtt} of the $K$-repetition \ac{harq} protocol is
\begin{equation}
    \label{eq:krep_rtt}
    T^{K_{rep}}_{RTT} = K_{rep} + 3 \text{ TTIs}.
\end{equation}
Therefore, the total latency of $m$ $K$-repetition transmissions is given by
\begin{equation}
    \label{eq:krep_latency}
    T^{K_{rep}}(m) = T_A + m T^{K_{rep}}_{RTT} = 1 + m (K_{rep} + 3) \text{ TTIs}.
\end{equation}

\section{System Analysis}
\label{sec:system_analysis}
The main requirement of \ac{urllc} applications is to reliably keep the user-plane latency below an application-dependent latency constraint. We begin this section by unambiguously defining what we mean by reliably and user-plane latency.
\begin{definition}[User-Plane Latency]
    User-plane latency is the time spent between the arrival of a packet to the \ac{ue}'s queue and the successful processing of an \ac{ack} signal received from the \ac{bs}.
\end{definition}
\begin{definition}[Latent Access Failure Probability Requirement]
    Latent access failure probability $\mathcal{P}_F(T \leq \tau)$, where $T$ is the user-plane latency and $\tau$ is the latency constraint, is the probability that the \ac{ue} data cannot be successfully decoded.
\end{definition}
Therefore, the \ac{qos} requirement of \ac{urllc} applications can be stated as
\begin{equation}
    \label{eq:urllc_qos}
    \mathcal{P}_F(T \leq \tau) \leq \epsilon,
\end{equation}
where $\tau$ is the latency constraint and $\epsilon$ is the minimum reliability, and both are application-dependent. Thus, to satisfy the \ac{qos} requirement, the probability that an \ac{ue} cannot transmit its data before $\tau$ must be bounded by $\epsilon$. Typically, $\tau$ varies between $1$ and $10$ ms and $\epsilon$ varies between $10^{-5}$ and $10^{-6}$ depending on the \ac{urllc} application.

Let $M$ be the maximum number of retransmissions under the latency constraint $\tau$. Moreover, notice that some of the \acp{ue} will transmit successfully earlier than others, and if the \ac{ue}'s transmission queue stays idle, the interference levels in distinct retransmissions are different. Therefore, the latent access failure probability is a function of the fraction of active users at the $m$-th retransmission ($\mathcal{A}_m$), the probability that the $m$-th retransmission is successful ($\mathcal{P}_m$( and the maximum number of retransmissions ($M$), as given by \cite{Liu2020}
\begin{equation}
    \label{eq:lafp}
    \mathcal{P}_F(T \leq \tau) = 
    \begin{cases}
        1, &\text{, if } M = 0 \\
        1 - \underset{m = 1}{\overset{M}{\sum}} \mathcal{A}_m \mathcal{P}_m &\text{, if } M \geq 1,
    \end{cases}
\end{equation}
where $\mathcal{A}_m$ is
\begin{equation}
    \label{eq:active_ratio}
    \mathcal{A}_m = 
    \begin{cases}
        1, &\text{, if } m = 1 \\
        1 - \underset{i = 1}{\overset{m-1}{\sum}} \mathcal{A}_i \mathcal{P}_i &\text{, if } m \geq 2.
    \end{cases}
\end{equation}
Given the expressions for $\mathcal{P}_m$, the latent access failure probability is obtained by iteratively computing (\ref{eq:active_ratio}) and (\ref{eq:lafp}).

In the rest of this section, we derive closed-form expressions for $\mathcal{P}_m$ under the reactive and $K$-repetition \ac{harq} protocols, denoted by $\mathcal{P}^{reac}_m$ and $\mathcal{P}^{Krep}_m$, respectively. To do so, we use stochastic geometric analysis to obtain the probability of success of a randomly chosen user $x_0$, herein the typical user. From Slivnyak's theorem \cite{Baccelli2009}, the performance of the typical user in an \ac{hppp} is representative of the average user's performance.

\subsection{Reactive \ac{harq}}
\label{sec:analysis_reactive}
The maximum number of \ac{harq} transmissions following the reactive \ac{harq} protocol with the delay constraint $\tau$ is given by
\begin{equation}
    \label{eq:reactive_max_tx}
    M^{reac} = \ceil{\frac{\tau - 1}{T^{reac}_{RTT}}} = \ceil{\frac{\tau - 1}{4}}.
\end{equation}
The first step in deriving an expression for the latent access failure probability is to obtain the probability that the $m$-th reactive retransmission is successful ($\mathcal{P}^{reac}_m$). 
Let $\Phi_I = \{x_i \lvert x_i \in \Phi_A \setminus \curly{x_0} \cap x_i \in LOS_m \}$ be the set of users interfering with the typical user's transmission on the $m$-th \ac{rtt}. Notice that due to the exponentially decreasing probability of a \ac{los} link with the increase in distance, $\Phi_I$ is a non-\ac{hppp} with density $\lambda_I(x) = \lambda_A \exp \paren{- \beta \norm{x}}$. The mean measure of $\Phi_I$, the average number of points in a given area, is obtained as
\ifCLASSOPTIONtwocolumn
\begin{eqnarray}
    \label{eq:mean_measure}
    \Lambda \paren{b(0, r)} &=& E \bracket{\Phi_I \paren{b(0, r)}} = \int_{\mathbb{R}^2} \lambda_I(x) dx \nonumber \\ &=& 2 \pi \lambda_A \int_0^r \exp \paren{- \beta r} r dr d\theta \nonumber \\ &=& \frac{2 \pi \lambda_A}{\beta^2} \bracket{1 - \exp \paren{- \beta r} \paren{1 + \beta r}}, \nonumber \\
\end{eqnarray}
\else
\begin{eqnarray}
    \label{eq:mean_measure}
    \Lambda \paren{b(0, r)} &=& E \bracket{\Phi_I \paren{b(0, r)}} = \int_{\mathbb{R}^2} \lambda_I(x) dx \nonumber \\ &=& 2 \pi \lambda_A \int_0^r \exp \paren{- \beta r} r dr d\theta = \frac{2 \pi \lambda_A}{\beta^2} \bracket{1 - \exp \paren{- \beta r} \paren{1 + \beta r}},
\end{eqnarray}
\fi
where $b(0,r)$ is a $2$-dimensional ball with radius $r$ that is centered at the origin. Now, let $N_m$ be a random variable denoting the number of users that interfere with the typical user on the $m$-th retransmission. From (\ref{eq:mean_measure}), the probability that there are $n$ interferers in the cell with radius R is derived as
\ifCLASSOPTIONtwocolumn
\begin{eqnarray}
    \label{eq:prob_interf}
    \mathbb{P}\paren{N_m = n} &=& \frac{\bracket{\mathcal{A}^{reac}_m \Lambda \paren{b(0, R)}}^n}{n!} \nonumber \\ &&\times \exp \paren{- \mathcal{A}^{reac}_m \Lambda \paren{b(0, R)}}. \nonumber \\
\end{eqnarray}
\else
\begin{equation}
    \label{eq:prob_interf}
    \mathbb{P}\paren{N_m = n} =\frac{\bracket{\mathcal{A}^{reac}_m \Lambda \paren{b(0, R)}}^n}{n!} \exp \paren{- \mathcal{A}^{reac}_m \Lambda \paren{b(0, R)}}.
\end{equation}
\fi
\begin{lemma}
    \label{lemma:ps_reac}
    If $K \gg 1$, the probability that the $m$-th reactive \ac{harq} retransmission of the typical user conditioned on the events that the typical user does not experience preamble collision, has a \ac{los} link and is affected by $n$ interferers can be approximated as
    \ifCLASSOPTIONtwocolumn
    \begin{eqnarray}
        \label{eq:ps_reac}
        &&\mathbb{P}\paren{\mathrm{SINR}_m \geq \gamma \left \lvert TX_0, N_m = n \right.} \nonumber \\ &\approx& \underset{n^\prime = 0}{\overset{n}{\sum}} {n \choose n^\prime} \paren{\frac{2}{K}}^{n^\prime} \paren{1 - \frac{2}{K}}^{n - n^\prime} \nonumber \\ 
        && \times \exp \paren{-\frac{\gamma}{\rho K}} \bracket{\frac{\tanh^{-1}\paren{\sqrt{\frac{\gamma}{1 + \gamma}}}}{\sqrt{\gamma \paren{1 + \gamma}}}}^{n^\prime}, \nonumber \\
    \end{eqnarray}
    \else
    \begin{eqnarray}
        \label{eq:ps_reac}
        &&\mathbb{P}\paren{\mathrm{SINR}_m \geq \gamma \left \lvert \bar{C}, x_0 \in LOS_m, N_m = n \right.} \nonumber \\ &\approx& \underset{n^\prime = 0}{\overset{n}{\sum}} {n \choose n^\prime} \paren{\frac{2}{K}}^{n^\prime} \paren{1 - \frac{2}{K}}^{n - n^\prime} \exp \paren{-\frac{\gamma}{\rho K}} \bracket{\frac{\tanh^{-1}\paren{\sqrt{\frac{\gamma}{1 + \gamma}}}}{\sqrt{\gamma \paren{1 + \gamma}}}}^{n^\prime},
    \end{eqnarray}
    \fi
    where $n^\prime$ is the number of interferers within the primary lobe of the beam directed at the typical user.
\end{lemma}

\begin{proof}
    See Appendix \ref{app:ps_reac}.
\end{proof}

After deriving the expressions for the probability of having $N_m$ users interfere with retransmission $m$ in (\ref{eq:prob_interf}) and the conditional probability of success obtained in Lemma \ref{lemma:ps_reac}, the success probability can be obtained as follows:

\begin{theorem}
\label{theorem:reac}
The probability that the $m$-th reactive \Ac{harq} retransmission is successfully decoded is
\ifCLASSOPTIONtwocolumn
\begin{eqnarray}
    \label{eq:reac_succ}
    &&\mathcal{P}_m^{reac} \nonumber \\ = &\underset{n = 0}{\overset{\infty}{\sum}}& \mathbb{P}(N_m = n) \mathbb{P}\paren{\left. \bar{C} \right.\lvert N_m = n} \nonumber\\ &&\times \mathbb{P}(x_0 \in LOS_m) \nonumber \\
    &\times\mathbb{P}&(\mathrm{SINR}_m \geq \gamma \left \lvert TX_0, N_m = n \right.),\nonumber \\
\end{eqnarray}
\else
\begin{eqnarray}
    \label{eq:reac_succ}
    \mathcal{P}_m^{reac} = &\underset{n = 0}{\overset{\infty}{\sum}}& \mathbb{P}(N_m = n) \mathbb{P}\paren{\left. \bar{C} \right.\lvert N_m = n} \nonumber\\ &&\mathbb{P}(x_0 \in LOS_m) \mathbb{P}(\mathrm{SINR}_m \geq \gamma \left \lvert \bar{C}, x_0 \in LOS_m, N_m = n \right.),
\end{eqnarray}
\fi
where $\mathbb{P}(N_m = n)$ is the probability that there are $n$ interferers in the cell and is given by (\ref{eq:prob_interf}). The probability of no preamble collision is given by
\begin{equation}
    \label{eq:non_col}
    \mathbb{P}\paren{\left. \bar{C} \right\lvert N_m = n} = \paren{1 - \frac{1}{Ns}}^n.
\end{equation}
And finally, the probability that the typical user has a \ac{los} link to the \ac{bs} is
\ifCLASSOPTIONtwocolumn
\begin{eqnarray}
    \label{eq:los_prob_m}
    \mathbb{P}\paren{x_i \in LOS_m} &=& \frac{1}{R} \int_0^R \exp \paren{ - \beta \norm{x_i}} d\norm{x_i} \nonumber \\ &=& \frac{1}{\beta R} \bracket{1 - \exp \paren{-\beta R}}.
\end{eqnarray}
\else
\begin{equation}
    \label{eq:los_prob_m}
    \mathbb{P}\paren{x_i \in LOS_m} = \frac{1}{R} \int_0^R \exp \paren{ - \beta \norm{x_i}} d\norm{x_i} = \frac{1}{\beta R} \bracket{1 - \exp \paren{-\beta R}}.
\end{equation}
\fi
\end{theorem}

\begin{proof}
The proof is straight forward if the conditional probability obtained in Lemma \ref{lemma:ps_reac} is averaged out.
\end{proof}

From the results of Theorem \ref{theorem:reac}, the latent access failure probability can be easily obtained by iteratively computing
\ifCLASSOPTIONtwocolumn
\begin{equation}
    \label{eq:reac_lafp}
    \mathcal{P}_F(T \leq \tau) = 
    \begin{cases}
        1 \text{, if } M^{reac} = 0 \\
        1 - \underset{m = 1}{\overset{M^{reac}}{\sum}} \mathcal{A}^{reac}_m \mathcal{P}^{reac}_m \text{,if } M^{reac} \geq 1.
    \end{cases}
\end{equation}
\else
\begin{equation}
    \label{eq:reac_lafp}
    \mathcal{P}_F(T \leq \tau) = 
    \begin{cases}
        1, &\text{, if } M^{reac} = 0 \\
        1 - \underset{m = 1}{\overset{M^{reac}}{\sum}} \mathcal{A}^{reac}_m \mathcal{P}^{reac}_m &\text{, if } M^{reac} \geq 1.
    \end{cases}
\end{equation}
\fi

\subsection{$K$-Repetition \ac{harq}}
\label{sec:krep_analysis}
In the $K$-repetition \ac{harq} system, the \ac{rtt} lasts from when the \ac{ue} transmits the first repetition until it receives the \ac{ack}/\ac{nack} feedback signal. Thus, under delay constraint $\tau$, the maximum number of retransmissions is
\begin{equation}
    M^{K_{rep}} = \ceil{\frac{\tau - 1}{T^{K_{rep}}_{RTT}}} = \ceil{\frac{\tau - 1}{K_{rep} + 3}}.
\end{equation}
Under the $K$-repetition \ac{harq}, the same data is repeated $K_{rep}$ times for every transmission attempt, and after the BS receives all the repetitions, it sends either an \ac{ack} or a \ac{nack} signal depending whether any of the repetitions sent in the transmission could be successfully decoded. Additionally, the \ac{ue} selects a new random subcarrier and preamble for the transmission of each distinct repetition. To obtain a closed-form expression for the latent access failure probability, we follow the same steps as were taken for the reactive \ac{harq} derivation.

\begin{lemma}
 \label{lemma:ps_krep}
    If $K \gg 1$, the probability that the $m$-th $K$-repetition \ac{harq} retransmission of the typical user, conditioned on the event that the typical user does not experience preamble collision, has a \ac{los} link and is affected by $n$ interferers can be approximated as 
    \ifCLASSOPTIONtwocolumn
    (\ref{eq:ps_krep}), shown at the top of next page,
    \begin{figure*}[ht!]
    \begin{eqnarray}
        \label{eq:ps_krep}
        \mathbb{P}\paren{\underset{l = 1}{\overset{K_{rep}}{\bigcup}}\mathrm{SINR}_{m,l} \geq \gamma \left \lvert \bar{C}, x_0 \in LOS_m, N_m = n \right.} \quad \nonumber \\ \approx \underset{l=1}{\overset{K_{rep}}{\sum}} {K_{rep} \choose l} \paren{-1}^{l+1} \underset{n^\prime = 0}{\overset{n}{\sum}} {n \choose n^\prime} \paren{\frac{2}{K}}^{n^\prime} \paren{1 - \frac{2}{K}}^{n - n^\prime} \nonumber \\ 
         \exp \paren{-\frac{l \gamma}{\rho K}} \bracket{\frac{\tanh^{-1}\paren{\sqrt{\frac{\gamma}{1 + \gamma}}}}{\sqrt{\gamma \paren{1 + \gamma}}}}^{l n^\prime}
    \end{eqnarray}
    \hrulefill
    \end{figure*}
    \else
    \begin{eqnarray}
        \label{eq:ps_krep}
        \mathbb{P}\paren{\underset{l = 1}{\overset{K_{rep}}{\bigcup}}\mathrm{SINR}_{m,l} \geq \gamma \left \lvert \bar{C}, x_0 \in LOS_m, N_m = n \right.} \quad \nonumber \\ \approx \underset{l=1}{\overset{K_{rep}}{\sum}} {K_{rep} \choose l} \paren{-1}^{l+1} \underset{n^\prime = 0}{\overset{n}{\sum}} {n \choose n^\prime} \paren{\frac{2}{K}}^{n^\prime} \paren{1 - \frac{2}{K}}^{n - n^\prime} \nonumber \\ 
         \exp \paren{-\frac{l \gamma}{\rho K}} \bracket{\frac{\tanh^{-1}\paren{\sqrt{\frac{\gamma}{1 + \gamma}}}}{\sqrt{\gamma \paren{1 + \gamma}}}}^{l n^\prime} ,
    \end{eqnarray}
    \fi
    where the double subscript $m,l$ indicates the $l$-th repetition of the $m$-th \ac{harq} retransmission attempt.
\end{lemma}

\begin{proof}
See Appendix \ref{app:ps_krep}.
\end{proof}

With the result from Lemma \ref{lemma:ps_krep}, the probability that the $m$-th retransmission attempt is successful can be obtained by averaging (\ref{eq:ps_krep}) over the conditional random variables.

\begin{theorem}
\label{theorem:krep}
The probability that the $m$-th $K$-repetition \Ac{harq} retransmission is successfully decoded is given by 
\ifCLASSOPTIONtwocolumn
(\ref{eq:krep_succ}), shown at the top of next page,
\begin{figure*}[ht!]
\begin{eqnarray}
    \label{eq:krep_succ}
    \mathcal{P}_m^{K_{rep}} = &\underset{n = 0}{\overset{\infty}{\sum}}& \mathbb{P}(N_m = n) \mathbb{P}\paren{\left. \bar{C} \right.\lvert N_m = n} \nonumber\\ &&\mathbb{P}(x_0 \in LOS) \mathbb{P}\paren{\underset{l = 1}{\overset{K_{rep}}{\bigcup}}\mathrm{SINR}_{m,l} \geq \gamma \left \lvert \bar{C}, x_0 \in LOS_m, N_m = n \right.},
\end{eqnarray}
\hrulefill
\end{figure*}
\else
\begin{eqnarray}
    \label{eq:krep_succ}
    \mathcal{P}_m^{K_{rep}} = &\underset{n = 0}{\overset{\infty}{\sum}}& \mathbb{P}(N_m = n) \mathbb{P}\paren{\left. \bar{C} \right.\lvert N_m = n} \nonumber\\ &&\mathbb{P}(x_0 \in LOS) \mathbb{P}\paren{\underset{l = 1}{\overset{K_{rep}}{\bigcup}}\mathrm{SINR}_{m,l} \geq \gamma \left \lvert \bar{C}, x_0 \in LOS_m, N_m = n \right.},
\end{eqnarray}
\fi
where the closed-form expression for $\mathbb{P}\paren{\underset{l = 1}{\overset{K_{rep}}{\bigcup}}\mathrm{SINR}_{m,l} \geq \gamma \left \lvert \bar{C}, x_0 \in LOS_m, N_m = n \right.}$ is derived on Lemma \ref{lemma:ps_krep}.
\end{theorem}

Given the analytical expression for the probability that the $m$-th $K$-repetition \ac{harq} retransmission is successfully received by the \ac{bs} in Theorem \ref{theorem:krep} and the fact that the probability that a randomly selected \ac{ue} is active can be computed from (\ref{eq:active_ratio}), the latent access failure probability is derived as
\ifCLASSOPTIONtwocolumn
\begin{equation}
    \label{eq:krep_lafp}
    \mathcal{P}_F(T \leq \tau) = 
    \begin{cases}
        1, \text{, if } M^{K_{rep}} = 0 \\
        1 - \underset{m = 1}{\overset{M^{K_{rep}}}{\sum}} \mathcal{A}^{K_{rep}}_m \mathcal{P}^{K_{rep}}_m \text{,}\\ \text{if } M^{K_{rep}} \geq 1.
    \end{cases}
\end{equation}
\else
\begin{equation}
    \label{eq:krep_lafp}
    \mathcal{P}_F(T \leq \tau) = 
    \begin{cases}
        1, &\text{, if } M^{K_{rep}} = 0 \\
        1 - \underset{m = 1}{\overset{M^{K_{rep}}}{\sum}} \mathcal{A}^{K_{rep}}_m \mathcal{P}^{K_{rep}}_m &\text{, if } M^{K_{rep}} \geq 1.
    \end{cases}
\end{equation}
\fi

\section{Numerical Results and Discussion}
\label{sec:numerical_results}
In this section, we report the results of Monte-Carlo simulations of the system model described in Section \ref{sec:system_model}. We use the simulation results to: a) validate the closed-form analytical approximations derived in Section \ref{sec:system_analysis} b) characterize the performance of the two \ac{harq} protocols in the \ac{mmwave} massive \ac{mimo} scenario and c) discuss the insights provided by the analytical results.

At the beginning of each simulation instance, the users' locations are generated according to an \ac{hppp} inside a cell with radius $R = 0.5$ km. At every \ac{tti}:

\begin{itemize}
    \item The channel gain between the \acp{ue} and the \ac{bs} located at the origin is generated as an exponential random variable with unit mean.
    \item All active \acp{ue} are determined to have either a \ac{los} or \ac{nlos} link according to the probability in (\ref{eq:los_prob_m}), with $\beta = 1$.
    \item All active \acp{ue} select a random subcarrier from one of the $N_S=48$ subcarriers available.
    \item All active \acp{ue} select a random preamble from one of the $N_P=64$ preambles available.
    \item The \ac{bs} checks all \acp{ue} with \ac{los} links on every subcarrier for preamble collision.
    \item The \ac{bs} computes the dot product between the signal received and the conjugate beam for all the \acp{ue} whose preambles have not collided. If the resulting \ac{sinr} is greater than $\gamma = -2$ dB, the transmission is successful, otherwise it fails.
    \item The \ac{bs} sends an \ac{ack} feedback signal to the \acp{ue} whose transmission was successful and a \ac{nack} feedback signal to those whose transmission attempt failed. As the main goal of this work is to characterize grant-free uplink performance, we assume that the feedback sent through the downlink channel is error free.
    \item All \acp{ue} move to a new location.
\end{itemize}

In accordance with \ac{3gpp} standards \cite{3gpp.38.101, 3gpp.38.211}, we consider a \ac{tti} mini-slot having a duration of $0.125$ ms and a subcarrier spacing of $60$ kHz, which is a configuration compatible with 5G NR \ac{fr2} operation, located in the \ac{mmwave} spectrum. We consider a noise figure of $-174$ dBm/Hz, a path loss exponent of $\alpha=4$ and a received power threshold of $\rho = -130 dBm$.

\ifCLASSOPTIONtwocolumn
\begin{figure*}
    \centering {
    \includegraphics[width=2\columnwidth]{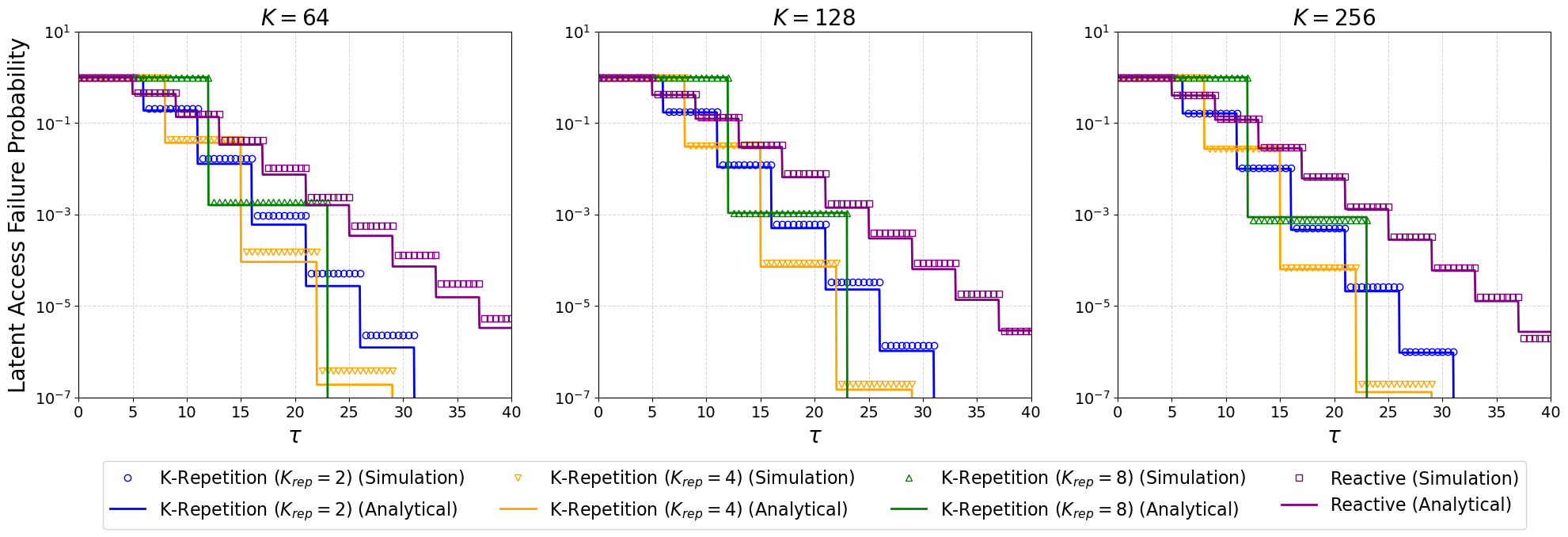}
    \caption{CCDF of the latent access failure probability for $\lambda_U = 1000$ \ac{ue}/$\text{km}^2$ for the reactive and $K$-repetition \ac{harq} protocols with $K_{rep} = 2, 4, 8$. The plots in the figure show the results for $K=64$, $K=128$ and $K=256$ antennas.}
    \label{fig:lafp_ccdf_lambda1000}
    }
\end{figure*}
\else
\begin{figure}
    \centering {
    \includegraphics[width=\columnwidth]{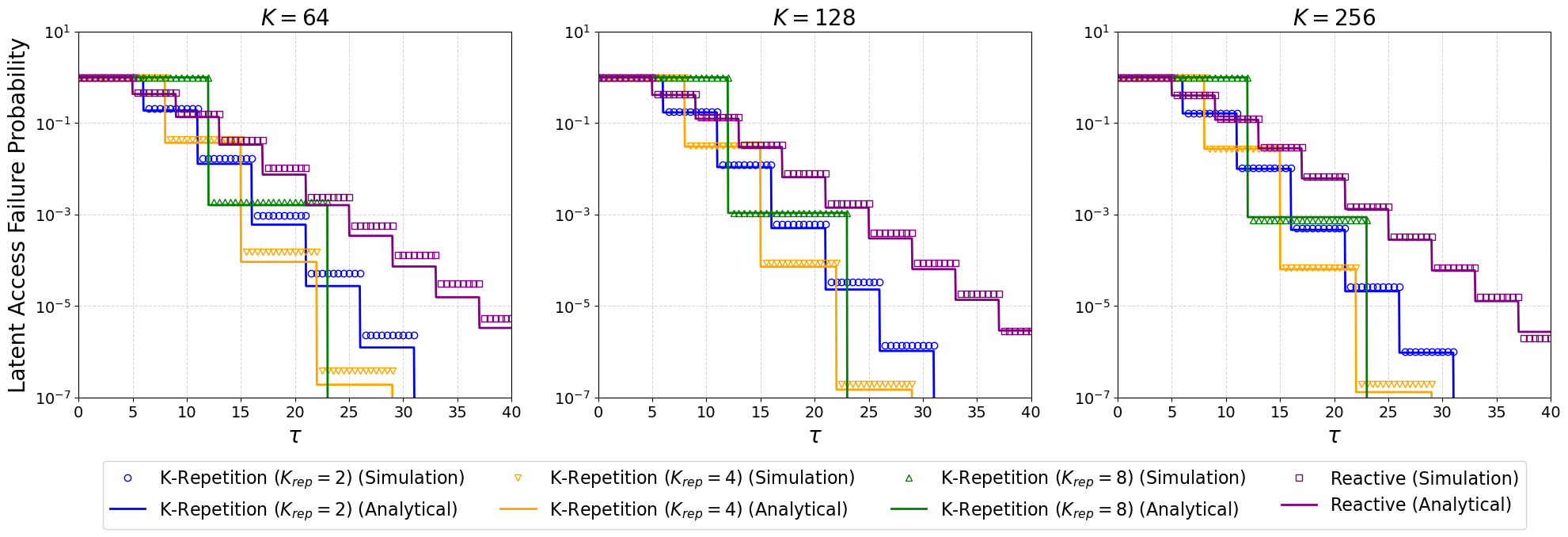}
    \caption{CCDF of the latent access failure probability for $\lambda_U = 1000$ \ac{ue}/$\text{km}^2$ for the reactive and $K$-repetition \ac{harq} protocols with $K_{rep} = 2, 4, 8$. The plots in the figure show the results for $K=64$, $K=128$ and $K=256$ antennas.}
    \label{fig:lafp_ccdf_lambda1000}
    }
\end{figure}
\fi
\ifCLASSOPTIONtwocolumn
\begin{figure*}
    \centering {
    \includegraphics[width=2\columnwidth]{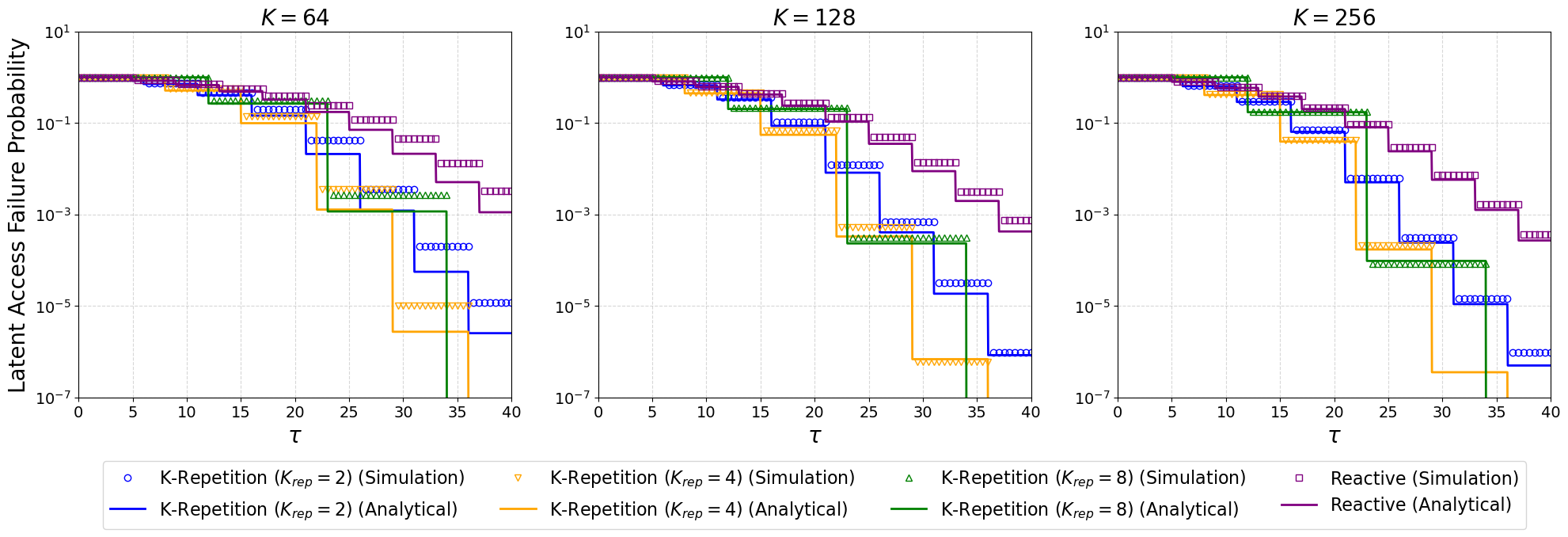}
    \caption{CCDF of the latent access failure probability for $\lambda_U = 5000$ \ac{ue}/$\text{km}^2$ for the reactive and $K$-repetition \ac{harq} protocols with $K_{rep} = 2, 4, 8$. The plots in the figure show the results for $K=64$, $K=128$ and $K=256$ antennas.}
    \label{fig:lafp_ccdf_lambda5000}
    }
\end{figure*}
\else
\begin{figure}
    \centering {
    \includegraphics[width=\columnwidth]{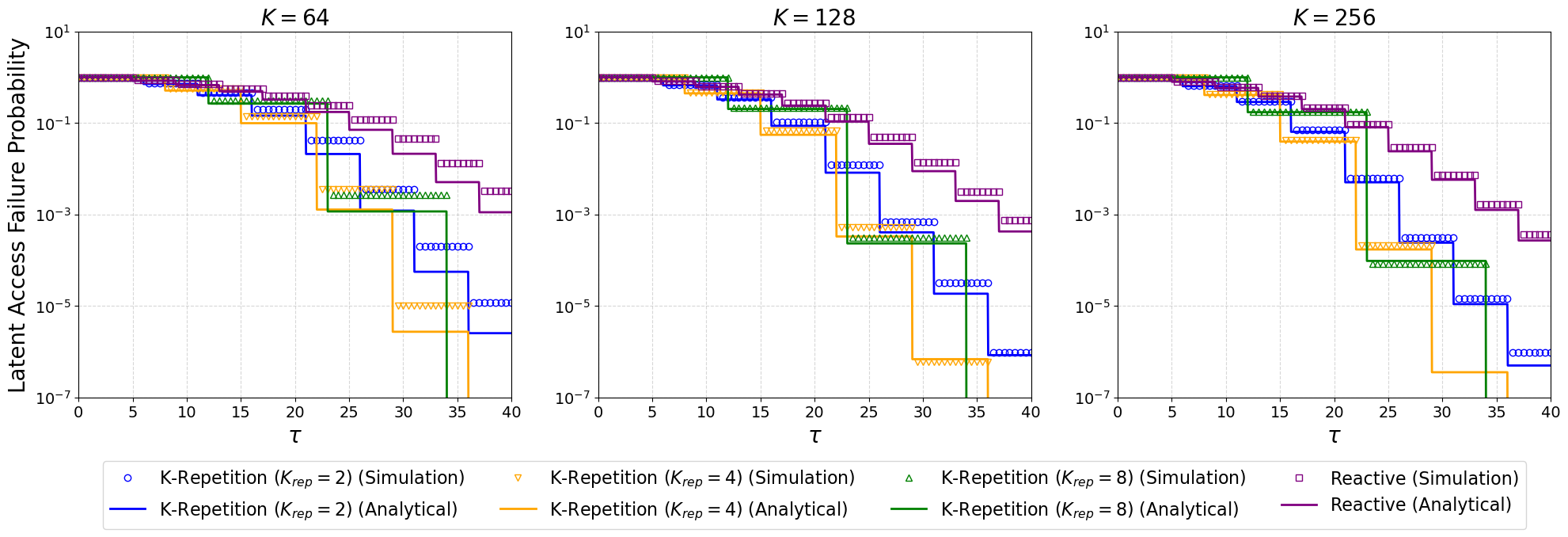}
    \caption{CCDF of the latent access failure probability for $\lambda_U = 5000$ \ac{ue}/$\text{km}^2$ for the reactive and $K$-repetition \ac{harq} protocols with $K_{rep} = 2, 4, 8$. The plots in the figure show the results for $K=64$, $K=128$ and $K=256$ antennas.}
    \label{fig:lafp_ccdf_lambda5000}
    }
\end{figure}
\fi

Figs. \ref{fig:lafp_ccdf_lambda1000} and \ref{fig:lafp_ccdf_lambda5000} show the \ac{ccdf} of the latent access probability for a user density of $\lambda_U = 1000$ \ac{ue}/$\text{km}^2$ and $\lambda_U = 5000$ \ac{ue}/$\text{km}^2$, respectively. The three plots in each figure display the performance for $64$, $128$ and $256$ antennas, from left to right. The behavior of the performance curves, where the latent access failure probability remains constant for a period of time and then drops on the following \ac{tti}, is due to the transmission propagation time on the uplink and the feedback, and the processing times. Fig. \ref{fig:lafp_ccdf_lambda1000} depicts the performance with a moderate \ac{ue} density scenario and shows that the reactive HARQ protocol is the best option for strict delay constraints, with $\tau \leq 6$ \acp{tti} ($0.725$ ms), as there is no time for any of the $K$-repetition configurations to finish their first round-trip. When the first and second round-trips for $K_{rep} = 2$ are completed, it has the best performance in $6 \leq \tau \leq 8$ \acp{tti} and $11 \leq \tau \leq 12$ \ac{tti} intervals. From this point on, the best performance is dominated by $K_{rep} = 4$ and $K_{rep} = 8$, with the best configuration being the one that has more completed round-trips in under $\tau$ \acp{tti}. A similar trend occurs with a higher user density as shown in Fig. \ref{fig:lafp_ccdf_lambda5000}.

Tables \ref{tab:lambda_1000} and \ref{tab:lambda_5000} show the reduction in the latent access failure probability upon increasing the number of antennas from $64$ to $256$. There is little improvement for a delay constraint of $1$ ms in either scenario. In applications with a moderate \ac{ue} density and a delay constraint of $2$ ms or more, we notice an average improvement of around $2$ across both \ac{harq} protocols investigated, while in applications with a higher user density, the failure probability is reduced by as much as $32$ times for $K_{rep} = 8$ repetitions and a delay constraint greater or equal to $3$ ms.
\begin{table}[ht!]
    \centering
    \begin{tabular}{c|ccc}
    \hline
        HARQ & $T \leq 1$ ms & $T \leq 2$ ms & $T \leq 3$ ms \\
    \hline
    \hline
        $K_{rep} = 2$ & 1.25 & 1.59 & 1.982 \\
        $K_{rep} = 4$ & 1 & 2.18 & 1.98 \\
        $K_{rep} = 8$ & 1 & 2.49 & - \\
        Reactive & 1.12 & 1.42 & 1.64 \\
    \end{tabular}
    \caption{Latent access failure probability reduction in increasing from $64$ to $256$ antennas when $\lambda_U = 1000$ \ac{ue}/$\text{km}^2$}
    \label{tab:lambda_1000}
\end{table}
\begin{table}[ht!]
    \centering
    \begin{tabular}{c|ccc}
    \hline
        HARQ & $T \leq 2$ ms & $T \leq 3$ ms & $T \leq 4$ ms \\
    \hline
    \hline
        $K_{rep} = 2$ & 1.54 & 6.81 & 13.69 \\
        $K_{rep} = 4$ & 3.25 & 17.13 & - \\
        $K_{rep} = 8$ & 1.75 & 32.51 & 32.51 \\
        Reactive & 1.40 & 2.60 & 6.18 \\
    \end{tabular}
    \caption{Latent access failure probability reduction in increasing from $64$ to $256$ antennas when $\lambda_U = 5000$ \ac{ue}/$\text{km}^2$}
    \label{tab:lambda_5000}
\end{table}
Nonetheless, notice from Figs. \ref{fig:lafp_ccdf_lambda1000} and \ref{fig:lafp_ccdf_lambda5000} that increasing the number of antennas from $128$ to $256$ does not change the latency performance significantly. Additionally, the increase in the number of antennas has a larger impact on performance for the higher user density scenario shown in Fig. \ref{fig:lafp_ccdf_lambda5000} than for the moderate density one in Fig. \ref{fig:lafp_ccdf_lambda1000}. Later in this section, we discuss why this happens and how to possibly address it. 

As the approximation used to derive the results in Section \ref{sec:system_analysis} relies on $K \gg 1$, there is a gap between the analytical and simulation results when $K = 64$ as, in this regime, the value of $F_K(x)$, and consequently the interference, outside the main lobe are no longer negligible in comparison to the gain on the main lobe.

\ifCLASSOPTIONtwocolumn
\begin{figure*}
    \centering {
    \includegraphics[width=2\columnwidth]{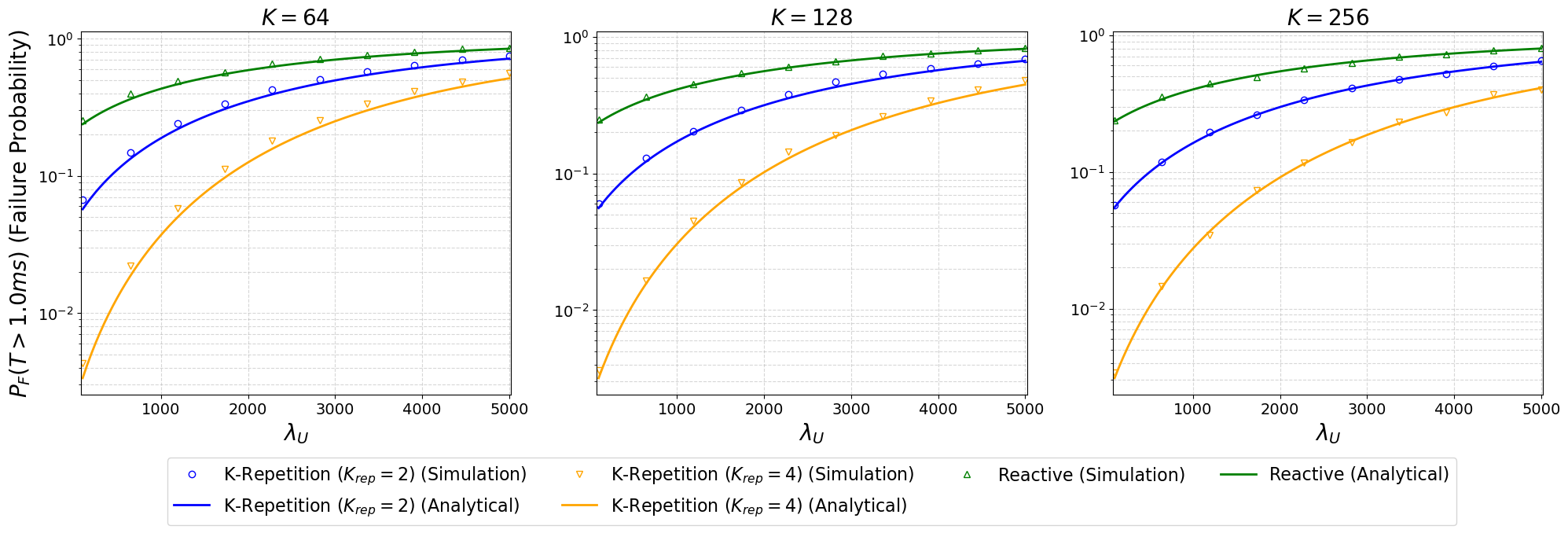}
    \caption{The probability that an \ac{ue} fails to transmit its packet under $\tau = 1$ ms for an user density ranging from $100$ \ac{ue}/$\text{km}^2$. The plots show the results for $K=64$, $K=128$ and $K=256$ antennas, respectively.}
    \label{fig:reliability_1ms}
    }
\end{figure*}
\else
\begin{figure}
    \centering {
    \includegraphics[width=\columnwidth]{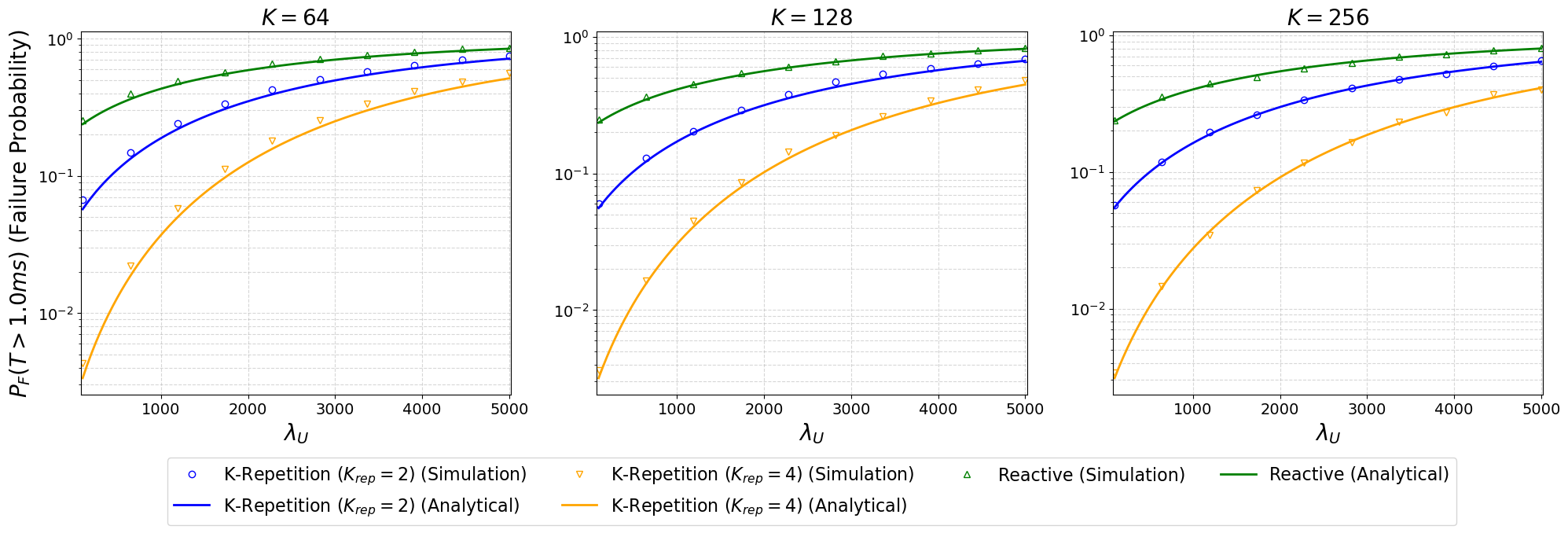}
    \caption{The probability that an \ac{ue} fails to transmit its packet under $\tau = 1$ ms for an user density ranging from $100$ \ac{ue}/$\text{km}^2$. The plots show the results for $K=64$, $K=128$ and $K=256$ antennas, respectively.}
    \label{fig:reliability_1ms}
    }
\end{figure}
\fi

\ifCLASSOPTIONtwocolumn
\begin{figure*}
    \centering {
    \includegraphics[width=2\columnwidth]{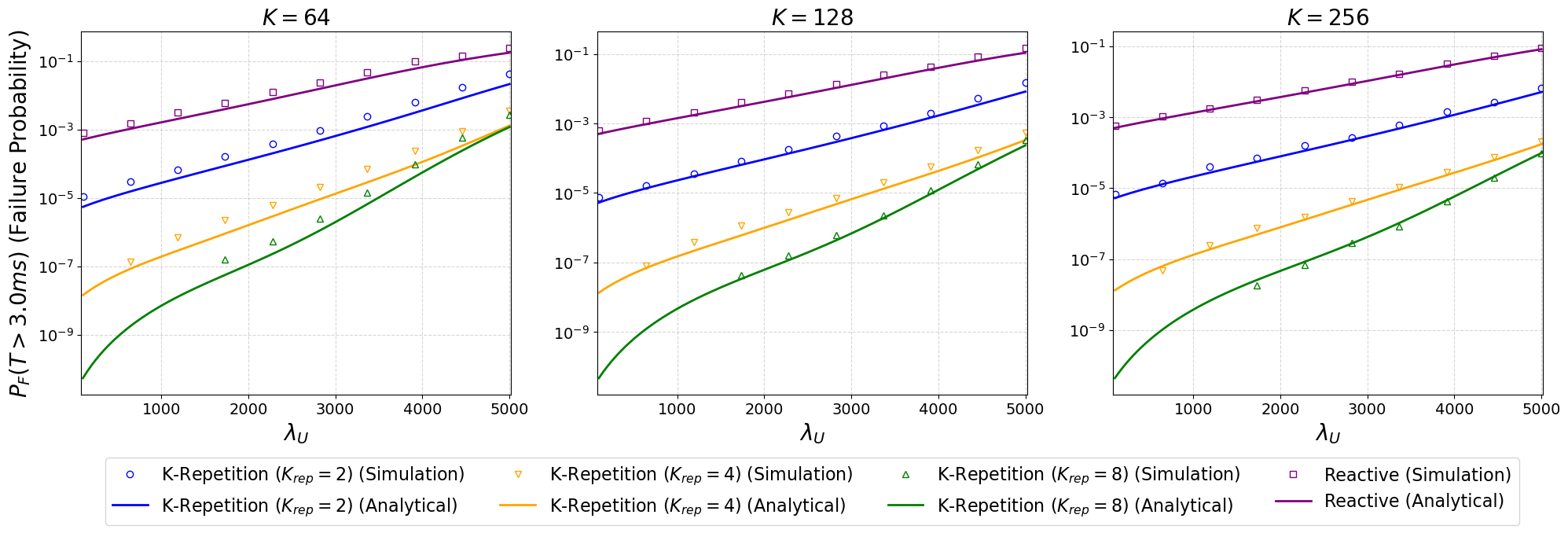}
    \caption{The probability that an \ac{ue} fails to transmit its packet under $\tau = 3$ ms for an user density ranging from $100$ \ac{ue}/$\text{km}^2$. The plots show the results for $K=64$, $K=128$ and $K=256$ antennas, respectively.}
    \label{fig:reliability_3ms}
    }
\end{figure*}
\else
\begin{figure}
    \centering {
    \includegraphics[width=\columnwidth]{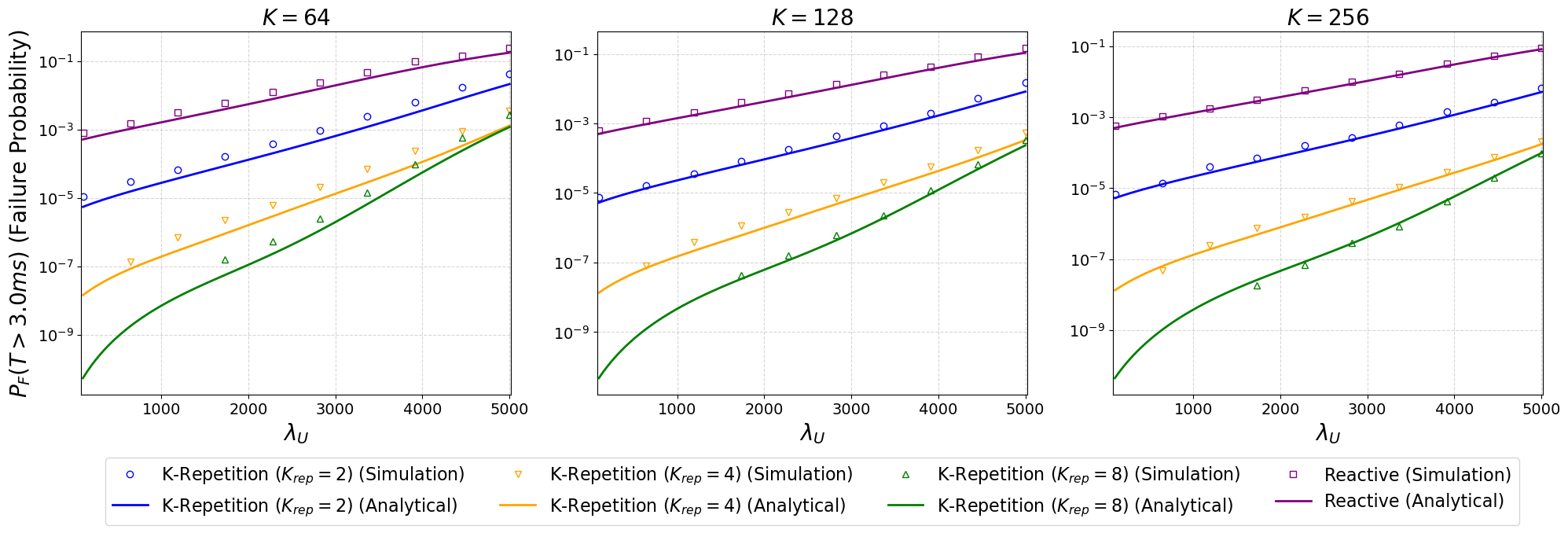}
    \caption{The probability that an \ac{ue} fails to transmit its packet under $\tau = 3$ ms for an user density ranging from $100$ \ac{ue}/$\text{km}^2$. The plots show the results for $K=64$, $K=128$ and $K=256$ antennas, respectively.}
    \label{fig:reliability_3ms}
    }
\end{figure}
\fi

Figs. \ref{fig:reliability_1ms} and \ref{fig:reliability_3ms} show how system reliability, i.e., in the probability of transmission failure under the latency constraint $\tau$, scales with an increase in user density for a latency constraint of $\tau = 1$ ms and $\tau=3$ ms, respectively. In both figures the user density ranges from $100$ to $5000$ \ac{ue}/$\text{km}^2$. Fig. \ref{fig:reliability_1ms} shows that that the combination of \ac{mmwave} and massive \ac{mimo} is not enough to satisfy the \ac{qos} requirement of \ac{urllc} applications with the stricter delay constraint (failure probability below $10^{-6}$). Also, in this latency range, the benefit from increasing the number of \ac{bs} antennas is rather small. For applications with less stringent delay constraints, shown in Fig. \ref{fig:reliability_3ms}, the $K$-repetition \ac{harq} protocol with $K_{rep} = 4$ and $K_{rep} = 8$ is able to support the \ac{urllc} \ac{qos} requirements. Table \ref{tab:qos_satisfy} depicts the highest \ac{ue} density that can be supported by each \ac{harq} and \ac{mimo} configuration. When $K_{rep} = 4$, increasing the number of \ac{bs} antennas from $64$ to $128$ increases the supported user density by $11\%$ and increasing the number from $128$ to $256$ increases it by only $7.5\%$. For $K_{rep} = 8$, those values are $14\%$ and $5\%$, respectively. It is worth noting that as long as the \ac{qos} constraints are satisfied, it is desirable to use the \ac{harq} configuration with the least number of repetitions as possible in order to save \ac{ue} power.
\begin{table*}[ht!]
    \centering
    \begin{tabular}{c|ccc}
    \hline
        HARQ & $K = 64$ antennas & $K= 128$ antennas & $K = 256$ antennas \\
    \hline
    \hline
        $K_{rep} = 2$ & - & - & - \\
        $K_{rep} = 4$ & 1800 \ac{ue}/$\text{km}^2$ & 2000 \ac{ue}/$\text{km}^2$ & 2150 \ac{ue}/$\text{km}^2$ \\
        $K_{rep} = 8$ & 2800 \ac{ue}/$\text{km}^2$ & 3200 \ac{ue}/$\text{km}^2$ & 3350 \ac{ue}/$\text{km}^2$ \\
        Reactive & - & - & - \\
    \end{tabular}
    \caption{User density supported (failure probability below $10^{-6}$) by each configuration.}
    \label{tab:qos_satisfy}
\end{table*}
\ifCLASSOPTIONtwocolumn
\begin{figure*}
    \centering {
    \includegraphics[width=1.7\columnwidth]{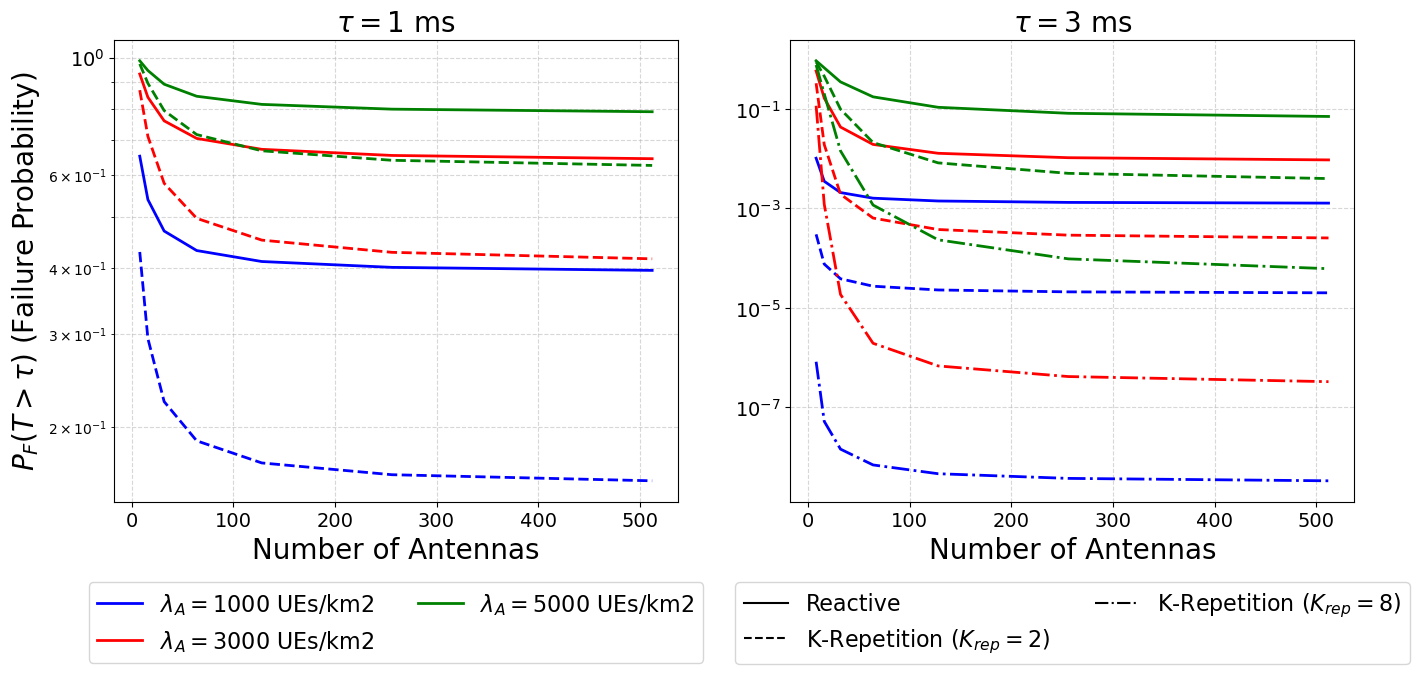}
    \caption{The impact of the number of antennas on the latent access failure probability. The leftmost plot shows results for delay constraint $\tau = 1$ ms, while the rightmost for $\tau = 3$ ms.}
    \label{fig:antenna_saturation}
    }
\end{figure*}
\else
\begin{figure}
    \centering {
    \includegraphics[width=.8\columnwidth]{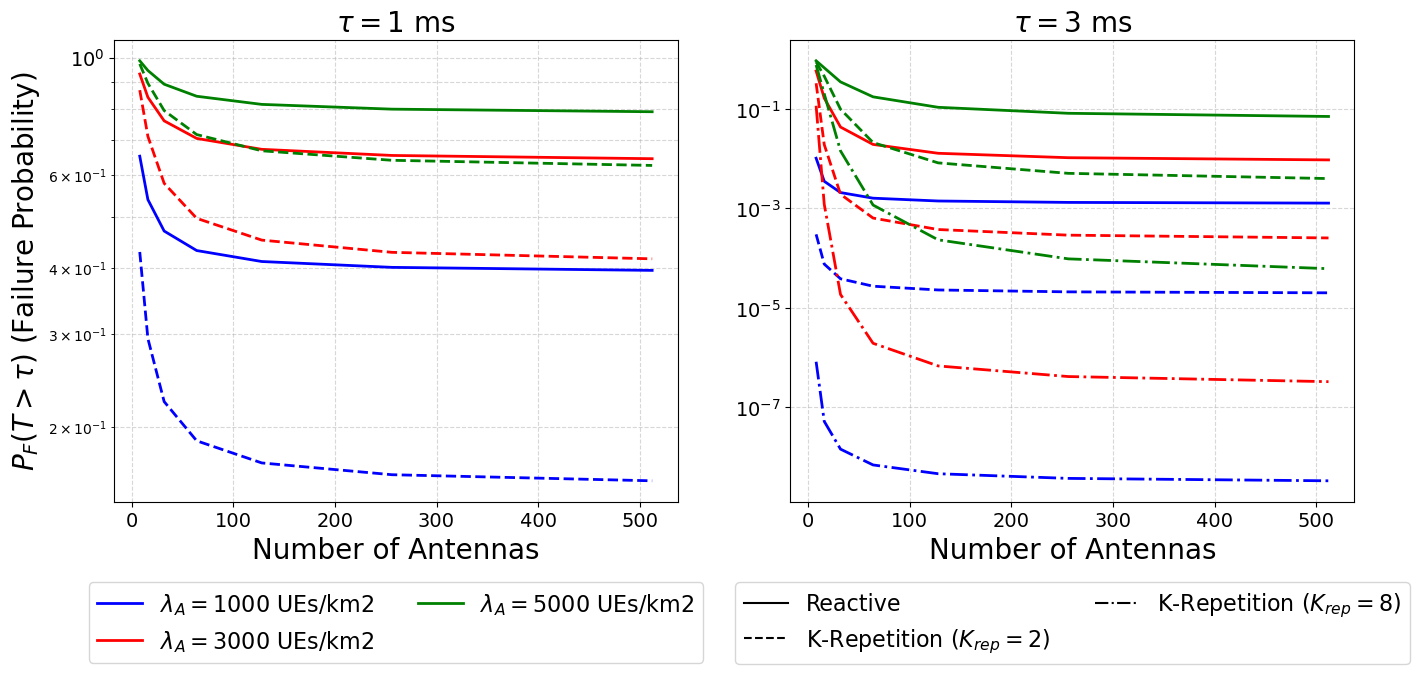}
    \caption{The impact of the number of antennas on the latent access failure probability. The leftmost plot shows results for delay constraint $\tau = 1$ ms, while the rightmost for $\tau = 3$ ms.}
    \label{fig:antenna_saturation}
    }
\end{figure}
\fi
In Fig. \ref{fig:antenna_saturation}, we show the impact of increasing the number of antennas on the failure probability for a delay constraint of $\tau = 1$ ms on the left and $\tau = 3$ ms on the right. From this figure, we can conclude that increasing the number of antennas beyond $100$ for the configuration under consideration ($R = 0.5$ km and $\beta = 1$) has a decreasing impact on latency performance. Moreover, the $K$-repetition \ac{harq} protocol benefits more from an increased number of antennas than the reactive \ac{harq} protocol does. Also, the plot on the right shows that the latency performance of applications with a moderate latency constraint ($\tau = 3$ ms) and higher user densities ($\lambda_U \geq 3000$ \ac{ue}/$\text{km}^2$) is greatly improved by changing from traditional \ac{mimo} to massive \ac{mimo}. This is explained by the capability of producing narrower receiver beams on systems with a higher number of antennas. Nevertheless, as of a certain point, the probability of having a \ac{los} link becomes the dominant bottleneck in reducing the latent access failure probability. As the \ac{los} link probability is unaffected by the number of antennas, another measure must be taken to further reduce the latency. In the system model formulated in this paper, one way to achieve this would be to increase the \ac{bs} deployment density, effectively decreasing the radius of the cells.

\section{Conclusions}
\label{sec:conclusions}
In this work, we formulated a model to analyze the latency and reliability of \ac{mmwave} massive \ac{mimo} \ac{urllc} applications using reactive and $K$-repetition \ac{harq} protocols. We used stochastic geometric spatiotemporal tools to derive closed-form approximations of the system's latent access failure probability. We validated the analytical results using Monte-Carlo simulations, identifying the limitations of our analytical results. Also, we investigated how the system's performance is impacted by the application's latency constraint, the density of \acp{ue} served by the system, and the number of antennas in the \ac{bs}. We concluded that:
\begin{itemize}
    \item Other than for extremely strict delay constraints ($\tau = 0.625$ ms), the $K$-repetition \ac{harq} protocol is a better choice.
    \item Increasing the number of \ac{bs} antennas from $64$ to $256$ \ac{bs} antennas can reduce the latent access failure probability by a factor of 32 for the cell configuration analyzed in the manuscript.
    \item Massive \ac{mimo}'s interference reduction capability significantly improves the reliability of systems with high user density and moderately improves the performance of systems with low user density.
    \item The increase in reliability from increasing the number of \ac{bs} antennas beyond $100$ is greatly reduced in the configuration investigated in Section \ref{sec:numerical_results}, as the probability of having a \ac{los} link between the \ac{ue} and the \ac{bs} becomes the main bottleneck.
    \item Under the configurations investigated in this manuscript, the system can support a \ac{ue} density as high as $3350$ \ac{ue}/$\text{km}^2$ for a \ac{urllc} application with latency and reliability constraints of $3$ ms and $10^{-6}$, respectively.
\end{itemize}
Overall, we can conclude that it is possible to increase the reliability of \ac{urllc} applications by using \ac{mmwave} massive \ac{mimo}, and when this technique is combined with selecting reasonable configuration parameters, these two techniques together can improve reliability under a strict latency constraint ($\tau = 1$ ms) and can satisfy \ac{urllc} \ac{qos} requirements under a less strict latency constraint ($\tau = 3$ ms).

\appendices

\section{Proof of Lemma \ref{lemma:ps_reac}}
\label{app:ps_reac}
The probability of success in (\ref{eq:ps_reac}) can be expanded to
\ifCLASSOPTIONtwocolumn
\begin{equation}
\begin{aligned}
    &\mathbb{P}\paren{\mathrm{SINR}_m \geq \gamma \left \lvert TX_0, N_m = n \right.} \\ 
    = &\mathbb{P} \curly{ \left.
        \abs{g_0}^2 \geq \frac{\gamma}{\rho K} \bracket{\sigma^2 + I} \right \rvert  TX_0, N_m = n 
    } \\
    =& \exp \paren{-\frac{\gamma \sigma^2}{\rho K}} \mathcal{L}_I\paren{s \lvert TX_0, N_m = n },
\end{aligned}
\end{equation}
\else
\begin{eqnarray}
    &&\mathbb{P}\paren{\mathrm{SINR}_m \geq \gamma \left \lvert \bar{C}, x_0 \in LOS_m, N_m = n \right.} \nonumber \\ &=& \mathbb{P} \curly{ \left.
        \abs{g_0}^2 \geq \frac{\gamma}{\rho K} \bracket{\sigma^2 + I} \right \rvert  \bar{C}, x_0 \in LOS_m, N_m = n 
    } \nonumber \\
    &=& \exp \paren{-\frac{\gamma \sigma^2}{\rho K}} \mathcal{L}_I\paren{s \lvert \bar{C}, x_0 \in LOS_m, N_m = n },
\end{eqnarray}
\fi
where $s = \frac{\gamma}{\rho K}$, $I = \underset{x_i \in \Phi_I}{\sum} \rho \abs{g_i}^2 F_K\paren{\frac{\pi}{2} \paren{\theta_0 - \theta_i}}$ is the interference on the typical user's transmission and $\mathcal{L}_I\paren{s \lvert \bar{C}, x_0 \in LOS_m, N_m = n }$ is the Laplace transform of the interference conditioned on the user not experiencing preamble collision, them having a \ac{los} path, and there being $n$ interferers in the cell.

The Laplace transform of the interference can be derived as
\ifCLASSOPTIONtwocolumn
 (\ref{eq:reac_lap_step1}), shown at the top of the next page.
\begin{figure*}[ht!]
\begin{eqnarray}
    \label{eq:reac_lap_step1}
    \mathcal{L}_I\paren{s \lvert TX_0, N_m = n } 
    &=& E_{g_i, \theta_i} \curly{\exp \bracket{
        -s \underset{x_i \in \Phi_I}{\sum} \rho \abs{g_i}^2 F_K\paren{\frac{\pi}{2} \paren{\theta_0 - \theta_i}}
    }} \nonumber \\
    &=& \underset{x_i \in \Phi_I}{\prod} E_{g_i, \theta_i} \curly{\exp \bracket{
        -s \rho \abs{g_i}^2 F_K\paren{\frac{\pi}{2} \paren{\theta_0 - \theta_i}}
    }} \nonumber \\
    &=& \underset{x_i \in \Phi_I}{\prod} E_{\theta_i} \bracket{
        \frac{1}{1 + s \rho F_K\paren{\frac{\pi}{2} \paren{\theta_0 - \theta_i}}}
    }.
\end{eqnarray}
\hrulefill
\end{figure*}
\else
\begin{eqnarray}
    \label{eq:reac_lap_step1}
    &&\mathcal{L}_I\paren{s \lvert \bar{C}, x_0 \in LOS_m, N_m = n } \nonumber \\
    &=& E_{g_i, \theta_i} \curly{\exp \bracket{
        -s \underset{x_i \in \Phi_I}{\sum} \rho \abs{g_i}^2 F_K\paren{\frac{\pi}{2} \paren{\theta_0 - \theta_i}}
    }} \nonumber \\
    &=& \underset{x_i \in \Phi_I}{\prod} E_{g_i, \theta_i} \curly{\exp \bracket{
        -s \rho \abs{g_i}^2 F_K\paren{\frac{\pi}{2} \paren{\theta_0 - \theta_i}}
    }} \nonumber \\
    &=& \underset{x_i \in \Phi_I}{\prod} E_{\theta_i} \bracket{
        \frac{1}{1 + s \rho F_K\paren{\frac{\pi}{2} \paren{\theta_0 - \theta_i}}}
    }.
\end{eqnarray}
\fi
Unfortunately, obtaining a closed-form expression for the expectation in (\ref{eq:reac_lap_step1}) is  not mathematically tractable. Therefore, we first obtain a suitable approximation to the Fejer kernel. Due to the Fejer kernel property in (\ref{eq:fejer_property}), as the number of antennas increases most of the energy is concentrated on the main lobe as shown in Fig. \ref{fig:fejer_kernel}. Hence, we choose to approximate it as
\ifCLASSOPTIONtwocolumn
\begin{equation}
    \label{eq:fejer_approx}
    F_K\paren{x} \approx f_K(x) =
    \begin{cases}
        - \frac{K^3 x^2}{4} + K \text{, if } x \in \curly{-\frac{2}{K}, \frac{2}{K}} \\
        0  \text{, otherwise}.
    \end{cases}
\end{equation}
\else
\begin{equation}
    \label{eq:fejer_approx}
    F_K\paren{x} \approx f_K(x) =
    \begin{cases}
        - \frac{K^3 x^2}{4} + K & \text{, if } x \in \curly{-\frac{2}{K}, \frac{2}{K}} \\
        0 & \text{, otherwise}.
    \end{cases}
\end{equation}
\fi
The quadratic approximation in (\ref{eq:fejer_approx}) renders the derivation of the expectation in (\ref{eq:reac_lap_step1}) tractable. Furthermore, it ensures that $f_K(x) = 0$ whenever $x \notin \paren{-\frac{2}{K}, \frac{2}{K}}$, i.e., the contributions of the signals arriving from directions outside of the main lobe to the interference is zero, and that $F_K\paren{0} = f_K(0) = K$. Hence,
\ifCLASSOPTIONtwocolumn
\begin{eqnarray}
    \label{eq:reac_lap_step2}
    \underset{x_i \in \Phi_I \cap \theta_i \in \paren{-\frac{2}{K}, \frac{2}{K}}}{\prod} E_{\theta_i} \bracket{
        \frac{1}{1 + s \rho F_K\paren{\frac{\pi}{2} \paren{\theta_0 - \theta_i}}}
    } \nonumber \\
    \overset{(a)}{=} \underset{x_i \in \Phi_I \cap \theta_i \in \paren{-\frac{2}{K}, \frac{2}{K}}}{\prod}
        \frac{\tanh^{-1}\paren{\sqrt{\frac{\gamma}{1 + \gamma}}}}{\sqrt{\gamma \paren{1 + \gamma}}}
    \nonumber \\
    \overset{(b)}{=} \underset{n^\prime = 0}{\overset{n}{\sum}} {n \choose n^\prime} \paren{\frac{2}{K}}^{n^\prime} \paren{1 - \frac{2}{ K}}^{n - n^\prime} \nonumber \\ \times \bracket{
        \frac{\tanh^{-1}\paren{\sqrt{\frac{\gamma}{1 + \gamma}}}}{\sqrt{\gamma \paren{1 + \gamma}}}
    }^{n^\prime}
\end{eqnarray}
\else
\begin{eqnarray}
    \label{eq:reac_lap_step2}
    && \underset{x_i \in \Phi_I \cap \theta_i \in \paren{-\frac{2}{K}, \frac{2}{K}}}{\prod} E_{\theta_i} \bracket{
        \frac{1}{1 + s \rho F_K\paren{\frac{\pi}{2} \paren{\theta_0 - \theta_i}}}
    } \nonumber \\
    &\overset{(a)}{=}& \underset{x_i \in \Phi_I \cap \theta_i \in \paren{-\frac{2}{K}, \frac{2}{K}}}{\prod}
        \frac{\tanh^{-1}\paren{\sqrt{\frac{\gamma}{1 + \gamma}}}}{\sqrt{\gamma \paren{1 + \gamma}}}
    \nonumber \\
    &\overset{(b)}{=}& \underset{n^\prime = 0}{\overset{n}{\sum}} {n \choose n^\prime} \paren{\frac{2}{K}}^{n^\prime} \paren{1 - \frac{2}{ K}}^{n - n^\prime} \bracket{
        \frac{\tanh^{-1}\paren{\sqrt{\frac{\gamma}{1 + \gamma}}}}{\sqrt{\gamma \paren{1 + \gamma}}}
    }^{n^\prime},
\end{eqnarray}
\fi
where $(a)$ is obtained from $\int \frac{1}{1 - x^2} dx = \tanh^{-1}(x)$. While step $(b)$ comes from the fact that the interferers' angles are uniformly distributed, given that there are $n$ interferers in the cell, the number of interferers within the typical user's main lobe direction follows a binomial distribution with $n^\prime \sim Binomial\paren{\frac{2}{K}}$. Thus, the conditional success probability can be obtained by summing the marginal distribution weighted by $n^\prime$'s \ac{pmf}. This completes the proof.

\section{Proof of Lemma \ref{lemma:ps_krep}}
\label{app:ps_krep}
In order for a $K$-repetition \ac{harq} transmission attempt to be successful at least one of the repetitions must be successfully decoded. Therefore, the probability that the $m$-th retransmission attempt is successful, conditioned on no preamble collisions, a \ac{los} path and $n$ interferers, can be obtained as the complement probability that all repetitions fail
\ifCLASSOPTIONtwocolumn
, as derived in (\ref{eq:krep_step1}), shown at the top of the next page,
\begin{figure*}[ht!]
\begin{eqnarray}
    \label{eq:krep_step1}
    \mathbb{P}\paren{\underset{l = 1}{\overset{K_{rep}}{\bigcup}}\mathrm{SINR}_{m,l} \geq \gamma \left \lvert TX_0, N_m = n \right.}
    &=& 1 - \mathbb{P} \paren{\underset{l = 1}{\overset{K_{rep}}{\bigcap}}\mathrm{SINR}_{m,l} < \gamma \left \lvert TX_0, N_m = n \right.} \nonumber \\
    &\overset{(a)}{=}& 1 - \underset{l = 1}{\overset{K_{rep}}{\prod}} \bracket{
        1 -  \mathbb{P} \paren{\mathrm{SINR}_{m,l} \geq \gamma \left \lvert TX_0, N_m = n \right.}
    } \nonumber \\
     &\overset{(b)}{=}& 1 - \bracket{1 -  \mathbb{P} \paren{\mathrm{SINR}_{m} \geq \gamma \left \lvert TX_0, N_m = n \right.}}^{K_{rep}} \nonumber \\
     &\overset{(c)}{=}& \underset{l=1}{\overset{K_{rep}}{\sum}} {K_{rep} \choose l} \paren{-1}^{l+1} \nonumber \\ 
     &&\mathbb{P} \paren{\mathrm{SINR}_{m} \geq \gamma \left \lvert TX_0, N_m = n \right.}^{l}
\end{eqnarray}
\hrulefill
\end{figure*}
\else
\begin{eqnarray}
    \label{eq:krep_step1}
    &&\mathbb{P}\paren{\underset{l = 1}{\overset{K_{rep}}{\bigcup}}\mathrm{SINR}_{m,l} \geq \gamma \left \lvert \bar{C}, x_0 \in LOS_m, N_m = n \right.} \quad \nonumber \\
    &=& 1 - \mathbb{P} \paren{\underset{l = 1}{\overset{K_{rep}}{\bigcap}}\mathrm{SINR}_{m,l} < \gamma \left \lvert \bar{C}, x_0 \in LOS_m, N_m = n \right.} \nonumber \\
    &\overset{(a)}{=}& 1 - \underset{l = 1}{\overset{K_{rep}}{\prod}} \bracket{
        1 -  \mathbb{P} \paren{\mathrm{SINR}_{m,l} \geq \gamma \left \lvert \bar{C}, x_0 \in LOS_m, N_m = n \right.}
    } \nonumber \\
     &\overset{(b)}{=}& 1 - \bracket{1 -  \mathbb{P} \paren{\mathrm{SINR}_{m} \geq \gamma \left \lvert \bar{C}, x_0 \in LOS_m, N_m = n \right.}}^{K_{rep}} \nonumber \\
     &\overset{(c)}{=}& \underset{l=1}{\overset{K_{rep}}{\sum}} {K_{rep} \choose l} \paren{-1}^{l+1} \mathbb{P} \paren{\mathrm{SINR}_{m} \geq \gamma \left \lvert \bar{C}, x_0 \in LOS_m, N_m = n \right.}^{l},
\end{eqnarray}
\fi
where step $(a)$ follows from the fact that the set of interferers is different from one repetition to the next, as a new subcarrier is randomly selected for every repetition by each \ac{ue}, making the \ac{sinr}s on distinct repetitions mutually independent. Also, as the \ac{sinr} of every repetition is affected by an interferer process having the same intensity, the probability of success of each repetition is equal, which justifies step $(b)$. Finally, step $(c)$ is obtained from the binomial expansion of the power term. If we work from (\ref{eq:krep_step1}) and follow the same steps derived in Appendix \ref{app:ps_reac}, we obtain (\ref{eq:ps_krep}), which completes the proof.

\ifCLASSOPTIONcaptionsoff
  \newpage
\fi

\bibliographystyle{IEEEtran}
\bibliography{IEEEabrv.bib,references.bib}{}









\end{document}